\catcode`\@=11
\newif\if@fewtab\@fewtabtrue
{\count255=\time\divide\count255 by 60
\xdef\hourmin{\number\count255}
\multiply\count255 by-60\advance\count255 by\time
\xdef\hourmin{\hourmin:\ifnum\count255<10 0\fi\the\count255}}
\def\ps@draft{\let\@mkboth\@gobbletwo
    \def\@oddfoot{\hbox to 7 cm{\tiny \versionno
       \hfil}\hskip -7cm\hfil\rm\thepage \hfil {\tiny\draftdate}}
    \def\@oddhead{}
    \def\@evenhead{}\let\@evenfoot\@oddfoot}
\def\draftdate{\number\month/\number\day/\number\year\ \ \ \hourmin }

\catcode`\@=12

\def\apo           {\mbox{\sc s}}
\def\apoi          {\mbox{\sc s}^{-1}}
\def\be            {\begin{equation}}
\newcommand\bee[5] {\begin{eqnarray} #5 \nonumber\\[-#1.#2em]~\\[#3.#4em]~\nonumber\end{eqnarray}}
\def\bico          {_{\F}}
\def\bbico         {_{\!F}}

\def\boti          {\,{\boxtimes}\,}
\def\bulkhh        {bulk handle Hopf algebra} 
\def\C             {{\ensuremath{\mathcal C}}}
\def\CbC           {\ensuremath{\hspace{.3pt}\overline{\mathcal C}\hspace{.6pt}{\boxtimes}%
                   \hspace{1.4pt}\mathcal C}\xspace}

\def\chii          {\raisebox{.15em}{$\chi$}}
\def\chirhh        {chiral handle Hopf algebra}
\newcommand\coen[1]{\int^{#1}\hspace*{-.13em}}
\def\complex       {{\ensuremath{\mathbbm C}}}
\def\CopC          {{\ensuremath{\mathcal C\op\hspace*{.8pt}{\times}\hspace*{1.5pt}\mathcal C}}}
\def\D             {{\ensuremath{\mathcal D}}}
\def\dim           {\mathrm{dim}}

\def\dimk          {\mathrm{dim}_\ko}
\def\drin          {f_Q}
\def\ee            {\end{equation}}

\def\EndC          {{\ensuremath{\mathrm{End}_\C}}}

\def\Endk          {{\ensuremath{\mathrm{End}_\ko}}}
\def\eps           {\varepsilon}
\def\F             {{\ensuremath{F}}}
\def\findim        {fi\-ni\-te-di\-men\-si\-o\-nal}
\def\Fo            {\ensuremath{\F_{\scriptscriptstyle\!C}}}
\def\H             {{\ensuremath L}}
\def\HBimod        {{\ensuremath{H\text{-}\mathrm{Bimod}}}}
\def\HMod          {{\ensuremath{H\text{-}\mathrm{Mod}}}}
\def\Hom           {{\ensuremath{\mathrm{Hom}}}}
\def\HomA          {{\ensuremath{\mathrm{Hom}_A}}}

\def\HomH          {{\ensuremath{\mathrm{Hom}_H}}}
\def\HomHH         {{\ensuremath{\mathrm{Hom}_{H|H}}}}
\def\Homk          {{\ensuremath{\mathrm{Hom}_\ko}}}
\def\Hs            {{\ensuremath{H^*_{}}}}
\def\Hss           {{H^*_{}}}
\def\id            {\mbox{\sl id}}
\def\idA           {\ensuremath{\id_A}}
\def\idH           {\ensuremath{\id_\H}}
\def\idHs          {\ensuremath{\id_{{H^{\phantom:}}^{\!\!*}}}}
\def\idsm          {\mbox{\footnotesize\sl id}}

\def\IJ            {{\ensuremath{\mathcal I}}}
\def\K             {{\ensuremath{K}}}
\def\ko            {{\ensuremath{\Bbbk}}}
\def\Mod           {\mbox{-Mod}}
\def\ohr           {\reflectbox{$\rho$}}
\def\one           {{\bf1}}
\def\op            {^{\mathrm{op}}}
\def\oti           {\,{\otimes}\,}
\def\Oti           {{\otimes}}
\def\otik          {\,{\otimes_\ko}\,}
\def\Qq            {\mathcal Q^{\rm l}}
\def\rep           {representation}
\def\To            {\,{\to}\,}
\def\twodim        {two-di\-men\-si\-o\-nal}
\def\V             {\ensuremath{\mathscr V}}
\def\wee           {*}  

\def\Zo            {\ensuremath{Z_{\scriptscriptstyle\!C}}}

\newcommand\Includepichtft[1] {{\begin{picture}(0,0)(0,0)
                   \scalebox{.38}{\includegraphics{imgs/pic_htft_#1.eps}}\end{picture}}}
\newcommand\eqpic[4]{\begin{eqnarray}
                   \begin{picture}(#2,#3){}\end{picture}\nonumber\\
                   \raisebox{-#3pt}{ \begin{picture}(#2,#3) #4 \end{picture} }
                   \label{#1} \\~\nonumber \end{eqnarray} }

\documentclass[12pt]{article}
\usepackage{latexsym, amsmath, amsthm, amsfonts, enumerate, amssymb, bbm, xspace, xypic }
\usepackage[mathscr]{eucal}
\usepackage{graphicx} \usepackage{rotating}

\newtheorem{theorem}{Theorem}
\newtheorem{lemma}[theorem]{Lemma}
\newtheorem{convention}[theorem]{Convention}
\newtheorem{proposition}[theorem]{Proposition}
\theoremstyle{definition}
\newtheorem{remark}[theorem]{Remark}

\begin{document}
\numberwithin{equation}{section}

\begin{flushright}
   {\sf ZMP-HH/12-2}\\
   {\sf Hamburger$\;$Beitr\"age$\;$zur$\;$Mathematik$\;$Nr.$\;$431}\\[2mm]
   January 2012
\end{flushright}
\vskip 4.5em
\begin{center} \Large\bf THE CARDY-CARTAN MODULAR INVARIANT
\\[2.4em]
  ~J\"urgen Fuchs\,$^{\,a}$,~
  ~Christoph Schweigert\,$^{\,b}$,~
  ~Carl Stigner\,$^{\,a}$
\end{center}

\vskip 9mm

\begin{center}\it$^a$
  Teoretisk fysik, \ Karlstads Universitet\\
  Universitetsgatan 21, \ S\,--\,651\,88\, Karlstad
\end{center}
\begin{center}\it$^b$
  Organisationseinheit Mathematik, \ Universit\"at Hamburg\\
  Bereich Algebra und Zahlentheorie\\
  Bundesstra\ss e 55, \ D\,--\,20\,146\, Hamburg
\end{center}
\vskip 5.3em

\noindent{\sc Abstract}
\\[3pt]
Using factorizable Hopf algebras, we construct modular invariant partition 
functions of charge conjugation, or Cardy, type as characters of coends in 
categories that share essential features with the ones appearing
in logarithmic CFT. The coefficients of such a partition function
are given by the Cartan matrix of the theory.

\newpage

\section{Introduction}

Partition functions are among the most basic quantities of a
quantum field theory. In large classes of \twodim\ rational conformal field theories,
the (torus) partition function is quite explicitly accessible: via the principle 
of holomorphic factorization, it can be written as a bilinear combination 
  \be
  Z = \sum_{i,j} Z_{i,j}\, \chii^\V_i \,{\otimes_\complex^{}}\, \chii^\V_j
  \label{Z}
  \ee
of the (finitely many) irreducible characters of the chiral symmetry algebra 
\V\ of the CFT, with non-negative integer coefficients $Z_{i,j}$.

By the uniqueness of the vacuum, for (semisimple) rational conformal field theories
one has $Z_{0,0} \,{=}\, 1$.  A further necessary condition on $Z \,{=}\, Z(\tau)$ 
is invariance under the action of the mapping class group of the torus on the 
characters. This condition, briefly referred to as \emph{modular invariance}, is
rather restrictive. It has been the starting point of several classification
programs. An important contribution by Max Kreuzer concerns a subclass of modular
invariant partition functions: those, in which the vacuum representation with character
$\chii^\V_0$ is only paired with irreducible \V-representations of a special type,
so-called simple currents \cite{scya6}, which are invertible objects in the representation
category of \V. All modular invariants that come in infinite families turn out to be 
of this type; to classify such invariants is thus an evident problem.

It is characteristic for Max Kreuzer's approach to problems in mathematical physics
that he did \emph{not} address \emph{this} problem. This is indeed most reasonable, because
the mentioned constraints imposed on the coefficients $Z_{i,j}$ are just necessary
conditions, but are far from sufficient. In fact, there are lots of examples of
combinations \eqref Z (with $Z_{0,0} \,{=}\, 1$) that are modular invariant, but are 
unphysical in the sense that they cannot be the torus partition function of any consistent
CFT. One reason for the insufficiency of the usual constraints is that the space of
bulk states, which are counted by $Z$, is not only a module over the tensor product
$\V \,{\otimes_\complex}\, \V$ of two copies of the chiral algebra, but has also 
more subtle properties, in particular it forms an algebra under the operator product,
with an invariant bilinear form derived from the two-point functions on the sphere.
To arrive at combinations \eqref Z that can be expected to be
compatible with such requirements as well, Max and Bert Schellekens
took an approach inspired by orbifold techniques. The outcome \cite{krSc} is a
general formula for $Z$ that is both beautiful and of tremendous use in applications.
Later on \cite{fuRs9} it could be proven rigorously that all modular invariants
covered by their formula constitute physical partition functions, forming part of a 
consistent collection of correlators with any number of field insertions at all genera.

Today concise mathematical formulations of the various conditions on partition
functions in rational CFT are available. The proper formalization of the
Moore-Seiberg data associated to a chiral algebra \V\ is the structure of a modular
tensor category on the representation category \C\ of \V. The bulk state space \F\ 
is then required to be a commutative Frobenius algebra in the Deligne product \CbC,
which formalizes the physical idea of pairing left- and right-moving chiral degrees 
of freedom. If the category
\C\ is semisimple it is easy to give an example of such an algebra:
  \be
  \F = \Fo := \bigoplus_{i\in\IJ}\, S_i^\vee \boxtimes S_i^{} \,,
  \label{FC}
  \ee
with $(S_i)_{i\in\IJ}$ representatives for the finitely many isomorphism classes
of simple objects of \C, i.e.\ of irreducible representations of the chiral 
algebra. The corresponding partition function $Z \,{=}\, \Zo$, called
the charge conjugation modular invariant, has coefficients
$Z_{i,j} \,{=}\, \delta_{i,\overline j}$. When compatible conformally invariant
boundary conditions of the CFT are considered, this solution
is also referred to as the `Cardy case'.

\medskip

What is relevant for capturing the Moore-Seiberg data is the category \C\
as an abstract category (with much additional structure), not its concrete 
realization as the \rep\ category of a chiral algebra \V. It is an old idea 
to rephrase these data by regarding the abstract category \C\ as the category 
\HMod\ of modules over a `quantum group', i.e.\ over a Hopf algebra
$H$ with an $R$-matrix. In this case the Deligne product \CbC\ can be realized 
by the category \HBimod\ of $H$-bimodules (alternatively, owing to modularity, 
as the equivalent category of Yetter-Drinfeld modules over $H$). The algebraic 
structure decribing the bulk algebra \F\ is then the dual of the Hopf algebra 
$H$ (with $H$ seen as a bimodule over itself). In the present note the 
categories \C\ and \CbC\ are treated in this spirit, i.e.\ are realized as 
categories of modules and bimodules over a quasitriangular Hopf algebra.

The novelty in our discussion is that we do \emph{not} require the category \C\ 
to be semisimple and are thereby transcending the Moore-Seiberg framework. 
Semisimplicity, meaning that any representation can be fully decomposed
into a finite direct sum of irreducible representations, arises in quantum physics
as a consequence of unitarity.
Still our motivation is entirely physical: the categories we are working
with are closely related to categories arising in logarithmic conformal field
theories, with applications ranging from condensed matter physics to string theory
(for a guide to the literature, see \cite[Sect.\,2]{fgst}). 
We summarize our main findings:

      \def\leftmargini{1.31em}~\\[-2.65em] \begin{itemize} 
\item
Even in the absence of semisimplicity there is a bulk algebra $\F \,{=}\, \Fo$ 
that generalizes \eqref{FC}. In particular it is modular invariant in the 
appropriate manner.

\item
The partition function for this bulk algebra \F\ can still be expressed as a
bilinear combination of the characters of simple objects of \C. Moreover, the
matrix $\Zo \,{=}\, (Z_{i,j})$ turns out to be a natural quantity associated with
the category \C: it is the \emph{Cartan matrix} of \C, which describes\,%
 \footnote{~Concerning the notion of Cartan matrix of an associative algebra,
 or of its abelian category of modules, see e.g.\ \cite[Ch.\,1.7]{BEns1} for a
 textbook reference, as well as \cite{loren}.}
how projective objects decompose into simple objects.
 \\[2pt]
Put very briefly:
 \\
The symbol $C$ not only stands for
\underline{\sl C}\hspace{1.8pt}harge conjugation and
\underline{\sl C}\hspace{1.3pt}ardy, but also for
\underline{\sl C}\hspace{1.3pt}artan.

\end{itemize}

\noindent
It is worth stressing that characters and partition functions, which count
states, do not distinguish between direct sums and non-trivial extensions
of representations. Thus when the underlying category is non-semisimple
they carry much less physical information than in the semisimple case.
It is thus crucial that we do not just obtain the partition function $Z$,
but even the bulk algebra $F$ that has $Z$ as its character.


\section{Summary of concepts and results}

In this section we formulate, in Theorems \ref{thm1} and \ref{thm2} below,
our main results and collect the relevant background information that is 
needed to appreciate them. The proofs, as well as a more detailed description
of various pertinent concepts, will be given in Section \ref{sec:details}.


\subsection{Factorizable Hopf algebras}

As already mentioned, we assume that we can realize the category \C, which for 
rational CFT is the modular tensor category that formalizes the Moore-Seiberg 
data, as the category \HMod\ of left modules over a Hopf algebra $H$. More 
precisely, $H$ comes endowed with additional structure, as stated in the 
following convention; for brevity 
we refer to such algebras as \emph{factorizable Hopf algebras}:
  
\begin{convention}\label{conv1}
Throughout this note, a \emph{factorizable Hopf algebra} is a \findim\
factorizable ribbon Hopf algebra over an algebraically closed field \ko\ of
characteristic zero.
\end{convention}

In the CFT setting, \ko\ is the field of complex numbers.
Let us summarize the meaning of the qualifications imposed on the \ko-vector
space $H$: That $H$ is a \emph{Hopf algebra} means that it is endowed with a
product $m$, unit $\eta$, coproduct $\Delta$, counit $\eps$ and antipode $\apo$,
such that $(H,m,\eta)$ is a unital associative algebra and
$(H,\Delta,\eps)$ is a counital coassociative coalgebra, with the coproduct
being an algebra morphism from $H$ to $H\oti H$, and with the antipode satisfying
$m\,{\circ}\,(\id_H\oti\apo)\,{\circ}\,\Delta \,{=}\, \eta\,{\circ}\,\eps 
\,{=}\, m\,{\circ}\,(\apo\oti\id_H)\,{\circ}\,\Delta$\,.
A \emph{quasitriangular} Hopf algebra is a Hopf algebra $H$ endowed with an
invertible element $R \,{\in}\, H\otik H$, called the \emph{R-matrix}, that 
intertwines the coproduct $\Delta$ and the opposite coproduct 
$\Delta^{\!\rm op} \,{=}\, \tau_{H,H}\,{\circ}\, \Delta$ and satisfies
 \footnote{~We identify $H$ with the space $\Homk(\ko,H)$ of linear maps from
 \ko\ to $H$.}
  \be
  (\Delta \oti \id_H) \circ R = R_{13}\cdot R_{23} \qquad{\rm and}\qquad
  (\id_H \oti \Delta) \circ R = R_{13}\cdot R_{12} \,.
  \ee
A \emph{ribbon} Hopf algebra is a quasitriangular Hopf algebra $H$ endowed with
a central invertible element $v\,{\in}\, H$, called the \emph{ribbon element}, that 
obeys
   \be
   \apo \circ v = v \,, \qquad \eps \circ v = 1 \qquad{\rm and}\qquad
   \Delta \circ v = (v\oti v) \cdot Q^{-1} \,,
   \ee
where $Q \,{\in}\, H\otik H$ is the \emph{monodromy matrix} $Q \,{=}\, R_{21}
\,{\cdot}\, R \,{\equiv}\, (\tau_{H,H}\,{\circ}\, R)\,{\cdot}\, R$.
(In the CFT context, the R-matrix contains information about the braiding,
while the eigenvalues of the action of the ribbon element on a module give
the exponentiated conformal weights.)
 \\
A \emph{factorizable} Hopf algebra is a quasitriangular Hopf algebra $H$ whose
monodromy matrix can be written as $Q \,{=}\, \sum_\ell h_\ell \oti k_\ell$
with $\{h_\ell\}$ and $\{k_\ell\}$ two vector space bases of $H$.

The Hopf algebras that are presently thought to be of relevance for classes of
logarithmic conformal field theories do not fully fit into our framework, but are
very close. They do not have an $R$-matrix, but still a factorizable monodromy
matrix (see e.g.\ \cite{fgst,naTs2}) or live in a more general category than 
the one of \findim\ \ko-vector spaces \cite{seTi4}.


\subsection{Modules and bimodules over factorizable Hopf algebras}

We denote by \HMod\ the category of left $H$-modules and by \HBimod\ the one of
$H$-bimodules. Both of them are finite tensor categories in the sense of \cite{etos},
and they have a ribbon structure, i.e.\ there are families of duality, braiding
and twist morphisms satisfying the usual axioms. In particular, 
the tensor product functor is exact in both arguments, and
the set $\IJ$ of isomorphism classes of simple objects is finite.
If $H$ is semisimple, then \HMod\ and \HBimod\ are semisimple modular tensor
categories, like the \rep\ categories of chiral algebras in rational CFT.
It is worth pointing out that the tensor product of \HBimod\ is \emph{not}
the one over $H$, for which the vector space underlying a tensor product bimodule
$X \,{\otimes}_H Y$ is a non-trivial quotient of the vector space tensor product
$X \otik Y$ (and for which only the structure of $H$ as an associative algebra
is needed), but rather uses explicitly that $H$ is a \emph{bi}algebra: it
is obtained by pulling back the natural $H\oti H$-bimodule structure on
$X \oti Y$ along the coproduct to the structure of an $H$-bimodule
(for more details see \cite[Sect.\,2.2]{fuSs3}).

The \emph{Deligne tensor product} of two locally finite \ko-linear abelian
categories \C\ and \D\ is a category $\C\boti\D$ together with a bifunctor
$\boxtimes\colon \C\,{\times}\,\D \To \C\boti\D$ that is right exact and 
\ko-linear in both variables and has a universal property by which, in short,
bifunctors from $\C\,{\times}\,\D$ become functors from $\C\boti\D$.
If $\C\,{\simeq}\,A\Mod$ and $\D\,{\simeq}\,B\Mod$ are categories of left modules
over associative algebras $A$ and $B$, respectively, then their Deligne product
$\C\boti\D$ is equivalent to $(A{\otimes}B\op)\Mod$ as a \ko-linear abelian 
category. In our case, where $A \,{=}\, B \,{=}\, H$ is a factorizable Hopf 
algebra, upon an appropriate choice of braiding on the bimodule category
this in fact extends to an equivalence
  \be
  \HBimod \,\simeq\, \overline{\HMod} \,\boxtimes\, \HMod
  \label{HBimod-HModHMod}
  \ee
of ribbon categories, where $\overline{\HMod}$ is \HMod\ with opposite braiding
and twist \cite[App.\,A.3]{fuSs3}.


\subsection{The bulk Frobenius algebra}

It has been shown in \cite{fuSs3} that for $\C \,{=}\, \HMod$ a natural candidate
for the bulk state space is the \emph{coregular bimodule} \F, i.e.\ the dual
space $\Homk(H,\ko)$ of $H$ endowed with the duals of the regular left and
right actions of $H$. For semisimple $H$, this $H$-bimodule decomposes 
into simple bimodules as in \eqref{FC}.

As an object of $\CbC \,{=}\, \HBimod$, \F\ has properties characteristic for the 
bulk state space of a CFT:

\begin{theorem}\label{thm1}
{\rm (i)} The maps
  \be
  \begin{array}{ll}
  m\bico := {\Delta^{\phantom:}}^{\!\!\!\wee_{}} \,, \qquad
  \eta\bico := \eps^\wee \,, \qquad
  \eps\bico := {\Lambda^{\phantom:}}^{\!\!\!\wee} \qquad{\rm and}\qquad
  \\[2mm]
  \Delta\bico := {[ (\id_H \oti (\lambda\,{\circ}\, m)) \,{\circ}\,
  (\id_H\oti\apo\oti\id_H) \,{\circ}\, (\Delta\oti\id_H) ]}^\wee
  \end{array}
  \label{def-Hb-Frobalgebra}
  \ee
$($with $\Lambda$ and $\lambda$ the integral and cointegral of $H$, respectively$)$
endow the coregular bimodule \F\ with the structure of a Frobenius algebra
$(\F,m\bico,\eta\bico,\Delta\bico,\eps\bico)$ in the category \HBimod.
\\[2pt]
{\rm (ii)} \F\ is commutative, cocommutative and symmetric and has trivial twist.
\end{theorem}

\noindent
For the proof, see Propositions 2.10 and 3.1, Theorem 4.4 and Remark 4.9 in
\cite{fuSs3}.  We call \F\ the \emph{bulk Frobenius algebra}.
If $H$ is semisimple, \F\ has the structure of a Lagrangian algebra 
in the sense of \cite[Def.\,4.6]{dmno}.

Owing to factorizability,
besides the equivalence \eqref{HBimod-HModHMod} there is also an equivalence
of ribbon categories between \HBimod\ and the Drinfeld center of \HMod, and
thus between \HBimod\ and the category of Yetter-Drinfeld modules over $H$.
Hence instead of working with $H$-bimodules we could equivalently work with
Yetter-Drinfeld modules over $H$. In that setting, the algebra \F\ arises as 
the so-called \cite{ffrs,davy20} \emph{full center} of the tensor unit of \HMod.


\subsection{The partition function}

Via the principle of holomorphic factorization, correlation functions in
full CFT are, at least for rational CFTs, elements in spaces of conformal
blocks of the associated chiral CFT. They must be invariant under actions 
of mapping class groups and obey sewing constraints. It is an obvious 
question whether solutions satisfying these conditions still exist when the 
theory is no longer rational so that the category \C\ is non-semisimple. In 
fact, non-trivial solutions to locality and crossing 
symmetry constraints on the sphere with a modular invariant spectrum have 
been found in \cite{gaKa3} (compare also \cite{garW2}).

In our setting, in which conformal blocks are specific morphism spaces 
$\HomH(-,-)$ in the category \HMod, while correlation
functions are elements of morphism spaces $\HomHH(-,-)$ in \HBimod, 
we are able to answer this question in the affirmative for any factorizable
Hopf algebra $H$, for the particular case of the torus partition function,
i.e.\ the zero-point correlator on the torus. The corresponding space of 
zero-point conformal blocks on the torus
turns out to be $\HomH(\H,\one)$, where $\one$ is the 
tensor unit of \HMod\ (given by the field \ko\ endowed with the trivial 
left $H$-action $\eps$) and $\H\,{\in}\,\HMod$ is a certain Hopf algebra 
internal to \HMod; \H, which is called the \chirhh, will be described 
in detail in Section \ref{sec:chirhh}. Similarly, the torus partition 
function itself is the character
  \be
  Z = \chii^\K_\F ~\in \HomHH(\K,\one)
  \label{Z=chiKF}
  \ee
of the bulk Frobenius algebra \F\ with respect to a  Hopf algebra \K\ internal
to \HBimod. Here $\one$ is now the tensor unit of \HBimod\ (again the field 
\ko, now endowed with trivial left and right $H$-actions); the \emph{\bulkhh} 
\K\ will be described in detail in Section \ref{sec:bulkhh}.
It has been shown in \cite{fuSs3} that the morphism
$\chii^\K_\F$ is modular invariant, with respect to the natural action of the 
modular group that comes from the action \cite{lyub6} of the modular group 
on the space $\HomH(\H,\one)$ of conformal blocks.

As we will see, \K\ can be canonically identified with $\H \,{\otimes}\, \H$.
Holomorphic factorization thus amounts to identifying $\chii^\K_\F$ with a
bilinear expression of basis elements of $\HomH(\H,\one)$. We can show that
this is indeed the case and, moreover, recognize the resulting coefficients
as natural quantities for the category \HMod:

\begin{theorem}\label{thm2}
The partition function {\rm \eqref{Z=chiKF}} can be chirally decomposed as
  \be
  Z = \sum_{i,j\in\IJ} c_{\overline i,j}^{}\, \chii^\H_i \otimes \chii^\H_j \,,
  \label{eq:thm2}
  \ee
where $\{\chii^\H_i \,|\, i\,{\in}\,\IJ\}$ are characters of \H-modules,
$\overline i \,{\in}\, \IJ$ is the label dual to $i$,
and $\,C \,{=}\, \big(c_{i,j}\big){}_{i,j\in\IJ}^{}$ is the Cartan matrix of \HMod.
\end{theorem}

\begin{remark}
(i)
The entries $c_{i,j}$ of the Cartan matrix are non-negative integers. In general,
$c_{0,0}$ is larger than 1, but
this is \emph{not} in contradiction with the uniqueness of the vacuum.
\\[2pt]
(ii) 
Unless $H$ is semisimple, the space $\HomH(\H,\one)$ is \emph{not} spanned by
characters alone. A complement is provided by so-called pseudo-characters 
(compare e.g.\ \cite{miya8,gaTi,arNa}). It is thus a non-trivial statement
that a decomposition of the form \eqref{eq:thm2} exists, irrespective of the 
precise values of the coefficients.
\\[2pt]
(iii) 
The result fits with predictions for the bulk state space of
certain logarithmic CFTs, the $(1,p)$ triplet models and WZW models with 
supergroup target spaces \cite{qusc4,garu2}.
\end{remark}

The decomposition \eqref{eq:thm2} is the main new result of this note.
It would be difficult to establish this relation directly, as it is hard
to describe the characters for the chiral and \bulkhh s sufficiently explicitly.
Instead, our idea of proof is to relate these characters to characters for
modules and bimodules over the underlying ordinary Hopf algebra $H$ and then
invoke classical results for the latter. 
In some more detail, we will proceed as follows.

\def\leftmargini{1.31em}~\\[-2.65em]\begin{enumerate}\addtolength{\itemsep}{-4pt}%
\item
Using general results about \findim\ associative algebras and their \rep s we
deduce the formula \eqref{A-bimod-char-2} for the character of a self-injective
algebra $A$ as a bimodule over itself.
\item
Using the fact that the bimodule structures of the regular and coregular bimodules 
$H$ and \F\ are intertwined by the Frobenius map, the character 
formula \eqref{A-bimod-char} is translated to the analogous formula 
\eqref{F-H-bimod-char} for the $H$-bimodule \F.
\item
We observe (Lemma \ref{Lemma:rhoHV}) that
the \chirhh\ \H\ acts via partial monodromy on any object of \HMod.
Based on this result we can show that
\H-characters are obtained from $H$-characters by composing them
with the Drinfeld map, see formula \eqref{chiiH=drinchiiH}.
\item
We obtain a similar expression \eqref{chiKX-fQi-fQ} of \K-characters in terms of
characters for the Hopf algebra $H\oti H\op$.
\item
We show (Proposition \ref{prop:chiKX}) that \K-characters can be written
as bilinear combinations of \H-characters.
When applied to the character of \F\ as a \K-module, together with the
previous results this yields the decomposition \eqref{eq:thm2}.
\end{enumerate}


\section{Details}\label{sec:details}

\subsection{Associative algebras and characters}

For $A \,{=}\, (A,m,\eta)$ a \findim\ (unital, associative) algebra over the 
field \ko, the \emph{character} $\chii_M^A$ of a left $A$-module $M \,{=}\, 
(M,\rho)$ is, by definition, the partial trace of the \rep\ morphism $\rho$. 
Here the trace is taken in the sense of linear maps, i.e.\ in the category of 
\findim\ \ko-vector spaces. Thus
  \be
  \chii_M^A := \mathrm{tr}_M(\rho) = \tilde d_M \circ
  (\rho \oti \id_{M^\vee}) \circ (\idA\oti b_M) ~\in \Hom(A,\ko) \,,
  \ee
where $b_M \,{\in}\, \Homk(\ko,M \otik M^*)$ is the (right) coevaluation 
and $\tilde d_M \,{\in}\, \Homk(M \otik M^*,\ko)$ the (left) evaluation. Now 
for \findim\ \ko-vector spaces the left and right dualities coincide, in the 
sense that $\tilde d_M$ can be expressed through the right evaluation $d_M
\,{\in}\, \Homk(M^* \otik M,\ko)$ as
  \be
  \tilde d_M = d_M \circ \tau_{M,M^*_{}}
  \ee
with $\tau$ the flip map, and analogously for the two coevaluations.

In the sequel we will make use of the graphical calculus for strict\,%
 \footnote{~That the tensor product of the categories of our interest is
 strictly associative can -- just like in many other situations
 in which associativity does not, a priori, hold on the nose -- be assumed
 by invoking the Coherence Theorem.
 }
ribbon categories.
In this pictorial notation, the two descriptions of the character are
  \eqpic{def_char} {210} {27} { \put(0,3){
  \put(0,36)    {$ \chii_M^A ~= $}
  \put(45,3){ \begin{picture}(0,0)(0,0)
      \scalebox{.38}{\includegraphics{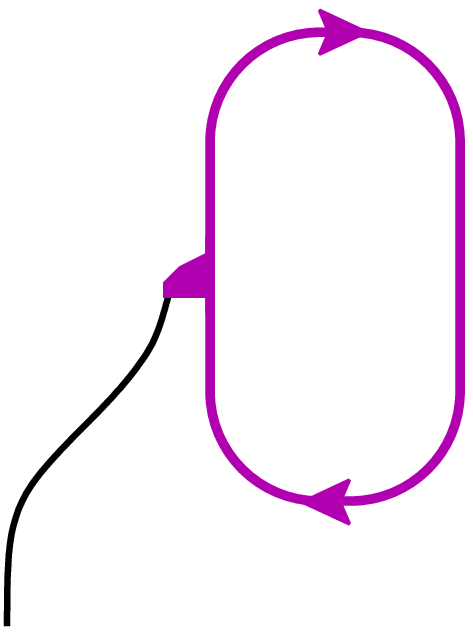}}\end{picture}
  \put(-3.4,-9) {\scriptsize$ A $}
  \put(14.6,59) {\scriptsize$ M $}
  }
  \put(125,36)  {$ = $}
  \put(150,3){ \begin{picture}(0,0)(0,0)
      \scalebox{.38}{\includegraphics{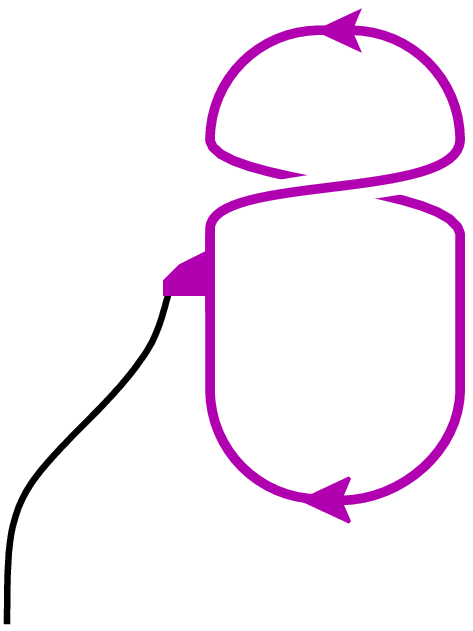}}\end{picture}
  \put(-3.4,-9) {\scriptsize$ A $}
  \put(51,60)   {\scriptsize$ M $}
  } } }
Characters are class functions, i.e.\ satisfy 
$\chii_M^A \,{\circ}\, m \,{=}\, \chii_M^A \,{\circ}\, m\op$;
$A$ is semisimple iff the space of class functions is already exhausted by linear
combinations of characters of $A$-modules.

Characters are additive under short exact sequences, i.e.\ for any short 
exact sequence $0 \To U \To W \To V \To 0$ of $A$-modules one has 
\cite[Sect.\,1.5]{loren} $\chii_W^A \,{=}\, \chii_U^A + \chii_V^A$. It 
follows that
  \be
  \chii_V^A = \sum_{i\in\IJ}\, [\,V \,{:}\, S_i\,] \, \chii_i^A \,,
  \label{chiAV}
  \ee
where $ \{ S_i \,|\, i\,{\in}\,\IJ \} $ is a full set of representatives of the 
isomorphism
classes of simple $A$-modules, $\chii_i^A \,{\equiv}\, \chii_{S_i^{}}^A$ is the
character of $S_i$ and $[\,V\,{:}\, S_i\,]$ is the multiplicity of $S_i$ in the
Jordan-H\"older series of $V$.
The simple $A$-modules $S_i$ are given by $P_i/J(A)\,P_i$, where $J(A)$ is the
Jacobson radical of $A$ and $P_i$ is the projective cover of $S_i$. The projective
modules $P_i$, in turn, form a full set of representatives of the isomorphism
classes of indecomposable projective $A$-modules and satisfy $P_i \,{=}\, A\,e_i$ 
with $ \{ e_i \,{\in}\, A \,|\, i\,{\in}\,\IJ \} $
a collection of primitive orthogonal idempotents.
With the same idempotents $e_i$, $Q_i \,{:=}\, e_i\,A$ are representatives  
for the isomorphism classes of indecomposable projective right $A$-modules.

The character of the projective module $P_i$ decomposes as
  \be
  \chii_{P_i}^A = \sum_{j\in\IJ} c_{i,j}^{} \, \chii_j^A
  \label{chiAPi}
  \ee
with coefficients
  \be
  c_{i,j} := [\,P_i \,{:}\, S_j \,] = \dimk(\HomA(P_i,P_j))
  \label{cij}
  \ee
for $i,j\,{\in}\,\IJ$. The matrix $C \,{=}\, \big(\,c_{i,j}\,\big)$, which only 
depends on the category $A\Mod$ as an abelian category, is called the
\emph{Cartan matrix} of $A$, or of the category $A\Mod$.

\smallskip

As a left module over itself, $A$ is projective and decomposes into 
indecomposable projective modules $P_i$ according to \cite[Satz\,G.10]{JAsc}
  \be
  _AA \,\cong\, \bigoplus_{i\in\IJ}\, P_i \otik \ko^{\dim(S_i)}_{} .
  \label{AA=}
  \ee
An analogous decomposition is valid for $A$ as a right module over itself.
The structure of $A$ as a bimodule over itself (with regular left and right 
actions) is, in general, much more complicated. Now the structure of 
an $A$-bimodule is equivalent to the one of a left $A{\otimes}A\op$-module;
accordingly, we define the character of $A$ as a bimodule over itself 
as the character of $A$ as an $A{\otimes}A\op$-module.

For any two \findim\ \ko-algebras $A$ and $B$, the Jacobson radical of
the tensor product algebra $A \oti B$ (i.e.\ the vector space $A \otik B$, 
endowed with unit map $\eta_A \oti \eta_B$ and product
$(m_A \oti m_B)\,{\circ}\, (\id_A \oti \tau_{A,B}\oti \id_B)$) satisfies 
$J(A{\otimes}B) \,{=}\, J(A) \otik B \,{+}\, A \otik J(B)$. Using that \ko\
is a field of characteristic zero, it follows that
a complete set of primitive orthogonal idempotents of $A \oti B$ is given by 
$\{e_i^A \otik e_j^B \,|\, i\,{\in}\,\IJ_A\,,\, j\,{\in}\,\IJ_B\}$, and complete 
sets of indecomposable projective and of simple $A{\otimes}B$-modules are
given by $\{P_i^A \otik P_j^B \,|\, i\,{\in}\,\IJ_A\,,\, j\,{\in}\,\IJ_B\}$
and by $\{S_i^A \otik S_j^B \,|\, i\,{\in}\,\IJ_A\,,\, j\,{\in}\,\IJ_B\}$, 
respectively (see e.g.\ \cite[Thm.\,(10.38)]{CUre1}).
Comparing with \eqref{chiAV}, it follows that the character of any $A\oti B$-
module $X$ can be written as a bilinear combination
  \be
  \chii_X^{A\otimes B}
  = \sum_{i\in\IJ_A,j\in\IJ_B} n_{i,j} \, \chii_i^A \oti \chii_j^B \,,
  \label{chiABX}
  \ee
where $\chii_k^A$ are the characters of simple left $A$-modules $S_k$, as 
above, and $\chii_l^B$ those of the simple left $B$-modules, and $n_{k,l}$ 
are non-negative integers.

Specializing to $B \,{=}\, A\op$ and $X \,{=}\, A$ (with regular actions), 
we can use that
  \be 
  \Hom_{\!A\otimes A\op_{}}\big( (A\oti A\op)\,(e_i\oti e_j) , A \big)
  \cong e_i\,A\,e_j \cong \HomA(Ae_i,Ae_j) = \HomA(P_i,P_j) \,,
  \ee
which implies that the entries \eqref{cij} of the Cartan matrix obey
  \be
  c_{i,j}
  = \dimk\big( \Hom_{\!A\otimes A\op_{}}( (A\Oti A\op)(e_i\Oti e_j),A) \big)
  = [\, A \,{:}\, S_i\otik T_j\,]
  \ee
with $T_k$ the simple quotients of the projective right $A$-modules $e_k\,A$. 
Using also that $\chii_M^{A\op_{}} {=}\, \chi^A_{M^*_{\phantom|}}$,
formula \eqref{chiAV} yields
  \be \begin{array}{ll}
  \chii_A^{A\otimes A\op_{}} \!\!&\displaystyle = \sum_{i,j\in\IJ}
  [\,A \,{:}\,S_i\otik T_j\,] \, \chii_{S_i\otimes_\ko^{}T_j}^{A\otimes A\op_{}}
  \\[-1.1em]\\ &\displaystyle
  = \sum_{i,j\in\IJ} c_{i,j} \, \chii_{S_i\otimes_\ko^{}T_j}^{A\otimes A\op_{}}
  = \sum_{i,j\in\IJ} c_{i,j} \, \chii_i^A \oti \chii_{T_j^{}}^{A\op_{}}
  = \sum_{i,j\in\IJ} c_{i,j} \, \chii_i^A \oti \chii_{S_j^*}^A
  \end{array} \label{A-bimod-char}
  \ee
as a linear map in $\Homk(A\oti A\op,\ko)$. 

Moreover, if $A$ is self-injective, then one has $T_k^{} \,{\cong}\, S_k^*$ as
right $A$-modules, so that \eqref{A-bimod-char} can be rewritten as
  \be
  \chii_A^{A\otimes A\op_{}}
  = \sum_{i,j\in\IJ} c_{i,j} \, \chii_i^A \oti \chii_j^A \,.
  \label{A-bimod-char-2}
  \ee


\subsection{Factorizable Hopf algebras}

Consider now the special case that $A \,{=}\, H$ is a \emph{factorizable Hopf} 
algebra in the sense of Convention \ref{conv1}, with coproduct $\Delta$, 
counit $\eps$ and antipode \apo. Then $H$ is in particular self-injective.
Thus by \eqref{A-bimod-char-2} the character of the \emph{regular $H$-bimodule},
i.e.\ the vector space $H$ together with the regular left and right actions,
is given by
  \be
  \chii_H^{H\otimes H\op}
  = \sum_{i,j\in\IJ} c_{i,j}^{} \, \chii_i^H \oti \chii_j^H \,.
  \label{H-bimod-char}
  \ee

Since the Hopf algebra $H$ is \findim, its antipode map $\apo$ is invertible, 
and there are one-dimensional spaces of left integrals $\Lambda \,{\in}\, H$ 
and of right cointegrals $\lambda \,{\in}\, \Hs$ \cite{laSw}. The composition 
$\lambda\,{\circ}\,\Lambda \,{\in}\, \ko$ is invertible (unless $\lambda$ or 
$\Lambda$ is zero), and we can and will choose the integral and cointegral
such that $\lambda\,{\circ}\,\Lambda \,{=}\, 1$. A factorizable Hopf algebra 
is unimodular \cite[Prop.\,3(c)]{radf13}, meaning that the left integral
$\Lambda$ is also a right integral, and this implies that 
$\apo\,{\circ}\,\Lambda \,{=}\, \Lambda$.

Next consider the \emph{coregular $H$-bimodule} \F, i.e.\ the vector space
$\Hs \,{=}\, \Homk(H,\ko)$ dual to $H$
endowed with the dual of the regular left and right actions; graphically,
   \eqpic{rhoHb,ohrHb} {140} {40} {
   \put(0,0)   {\begin{picture}(0,0)(0,0)
       \scalebox{.38}{\includegraphics{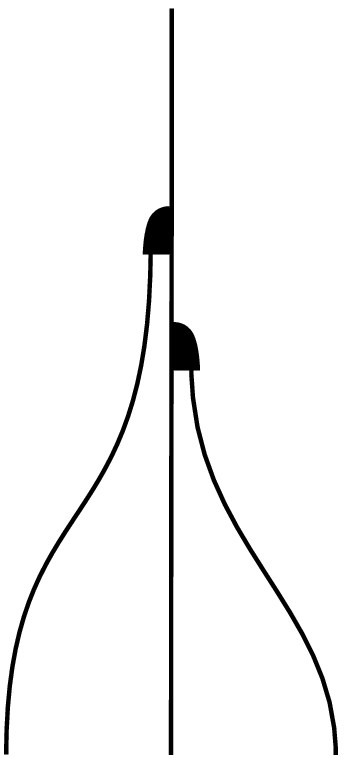}}\end{picture}
   \put(-4,-8.8)    {\scriptsize$ H $}
   \put(5,56)       {\scriptsize$ {\rho\bbico^{}} $}
   \put(13,-8.8)    {\scriptsize$ \Hss $}
   \put(14.3,84.9)  {\scriptsize$ \Hss $}
   \put(23,43)      {\scriptsize$ \ohr\bico $}
   \put(32,-8.8)    {\scriptsize$ H $}
   }
   \put(58,38)      {$ =$}
   \put(88,0) { \begin{picture}(0,0)(0,0)
       \scalebox{.38}{\includegraphics{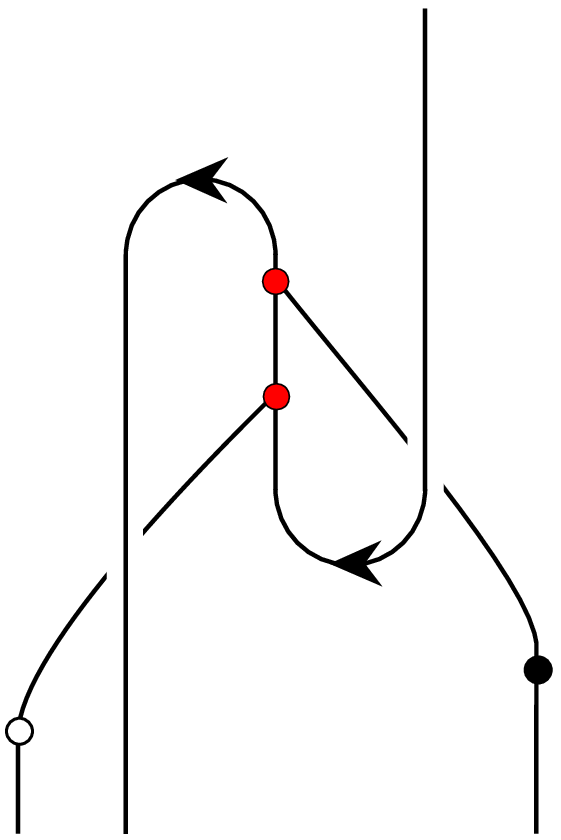}}\end{picture}
   \put(-5,9)       {\scriptsize$ \apo $}
   \put(-3,-8.8)    {\scriptsize$ H $}
   \put(10,-8.8)    {\scriptsize$ \Hss $}
   \put(42.3,93.3)  {\scriptsize$ \Hss $}
   \put(54,-8.8)    {\scriptsize$ H $}
   \put(61.6,16.2)  {\scriptsize$ \apoi $}
   } }
The bimodule \F\ is isomorphic as a bimodule to the regular bimodule.
Indeed, for any linear map $\mu \,{\in}\, \Hs$ the map
  \be
  \Phi_\mu := ((\mu\,{\circ}\, m) \oti \idHs) \circ (\apo \oti b_H)
  \ee
intertwines the regular and coregular left $H$-actions; if $\mu$ satisfies
$\mu \,{\circ}\, m \,{=}\, \mu \,{\circ}\, m\,{\circ}\,\tau_{H,H}\,
{\circ}\,(\idH\oti\apo^2)$
then $\Phi_\mu$ intertwines the regular and coregular right 
$H$-actions as well and is thus a morphism of $H$-bimodules. In particular we can 
take $\mu \,{=}\, \lambda$ to be the cointegral, in which case $\Psi \,{=}\, 
\Phi_\lambda$ is the \emph{Frobenius map} (see e.g.\ \cite{coWe6}); since 
the Frobenius map of a factorizable Hopf algebra is invertible (see e.g.\ 
\cite[App.\,A.2]{fuSc17}), 
it follows that indeed $H$ and \F\ are isomorphic as $H$-bimodules.

It follows that the decomposition \eqref{H-bimod-char} applies to the
coregular bimodule \F\ as well, i.e.
  \be
  \chii_\F^{H\otimes H\op}
  = \sum_{i,j\in\IJ} c_{i,j}^{} \, \chii_i^H \oti \chii_j^H \,.
  \label{F-H-bimod-char}
  \ee


\subsection{The \chirhh}\label{sec:chirhh}

To proceed, we introduce a certain Hopf algebra internal to the
category \HMod, the \emph{\chirhh} \H.
To this end we need a few notions from category theory. It is convenient to
formulate them first for \ko-linear abelian ribbon categories \C, and only 
later on specialize to the case that \C\ is the category $\HMod$ of left 
$H$-modules, which for any factorizable Hopf algebra $H$ belongs to this 
class of categories.

A \emph{dinatural transformation} $F\,{\Rightarrow}\,B$ from a functor
$F\colon\, \CopC\To\D$ to an object $B\,{\in}\,\D$ is a family of morphisms 
$\varphi \,{=}\, \{ \varphi_U\colon F(U,U)\To B \}_{\!U\in\C}^{}$ such that
the diagram
  \bee4010{
  \xymatrix @R+8pt{ & F(V,U) \ar^{F(\idsm_V,f)}[dr]\ar_{F(f,\idsm_U)}[dl]\\
  F(U,U) \ar^{\varphi_U^{}}[dr] && F(V,V)\ar^{\varphi_V^{}}[dl] \\ & B\,
  } }
commutes for all $f\,{\in}\,\Hom(U,V)$.
A \emph{coend} $(C,\iota)$ for a functor $F\colon \CopC\To\D$
is an object $C\,{\in}\,\D$ together with a dinatural transformation $\iota$ that
has the universal property that for any dinatural transformation
$\varphi\colon F\,{\Rightarrow}\,B$ there is a unique morphism
$\kappa \,{\in}\, \Hom_\D(C,B)$ such that $\varphi_U \,{=}\, \kappa \,{\circ}\, 
\iota_U$ for all objects $U$ of \C. Coends are unique up to unique isomorphism. 
The finiteness properties of the categories we are working with guarantee the 
existence of all coends we need.

Several different coends turn out to be of interest to us. The one relevant for
us now is the coend $\H \,{:=}\, \coen U U^\vee \oti U$ of the functor from \CopC\
to \C\ that acts on objects as $(U,V) \,{\mapsto}\, U^\vee \oti V$.
As shown in \cite{maji25,lyub8}, \H\ has a natural structure of a Hopf algebra
$(\H,m_\H,\eta_\H,\Delta_\H,\eps_\H,\apo_\H)$ internal to \C; its structural
morphisms are given by
  \be \begin{array}{ll}
  m_\H \circ (i_U\oti i_V)
  := i_{V\otimes U} \circ (\gamma_{U,V}\oti\id_{V\otimes U})
  \circ (\id_{U^\vee} \oti c_{U,V^\vee\otimes V}) \,,
  \qquad
  \eta_\H := i_\one \,,
  \\[-1.3em]\\[4mm]
  \Delta_\H \circ i_U := (i_U\oti i_U) \circ (\id_{U^\vee} \oti b_U \oti \id_U) \,,
  \qquad
  \eps_\H \circ i_U := d_U \,,
  \\[-1.3em]\\[4mm]
  \apo_\H \circ i_U := (d_U \oti i_{U^\vee}) \circ (\id_{U^\vee} \oti 
  c_{U^{\vee\vee},U} \oti \id_{U^\vee}) \circ (b_{U^\vee} \oti c_{U^\vee,U}) \,,
  \end{array} \label{H_hopf}
  \ee
where $\gamma_{U,V}^{}$ is the canonical identification of $U^\vee\oti V^\vee$
with $(V\oti U)^\vee$.
(Here it is used that a morphism $f$ with domain the coend \H\ is uniquely
determined by the dinatural family $\{ f\,{\circ}\, i_U\}$ of morphisms.)
In graphical notation,\,%
 \footnote{~The picture \eqref{coend_m}, as well as \eqref{QHX}, \eqref{Qrep} and
 \eqref{Qrep2} below, describe morphisms in the monoidal category \C, which
 (generically) is genuinely braided, i.e.\ over- and underbraiding are
 different morphisms. In the pictures this is indicated by labeling the
 braiding explicitly with the symbol $c$.
 In contrast, the other pictures we display refer to the category of \findim\
 \ko-vector spaces, for which the braiding is just the flip map $\tau$, so that
 over- and underbraiding coincide. Note that $\tau$ is just a linear map,
 rather than a morphism of \HMod\ or \HBimod; nevertheless the maps described
 in those pictures \emph{are} morphisms in the relevant categories.}
  \eqpic{coend_m} {390} {108} {
   \put(0,113) {
  \put(0,0)   {\Includepichtft{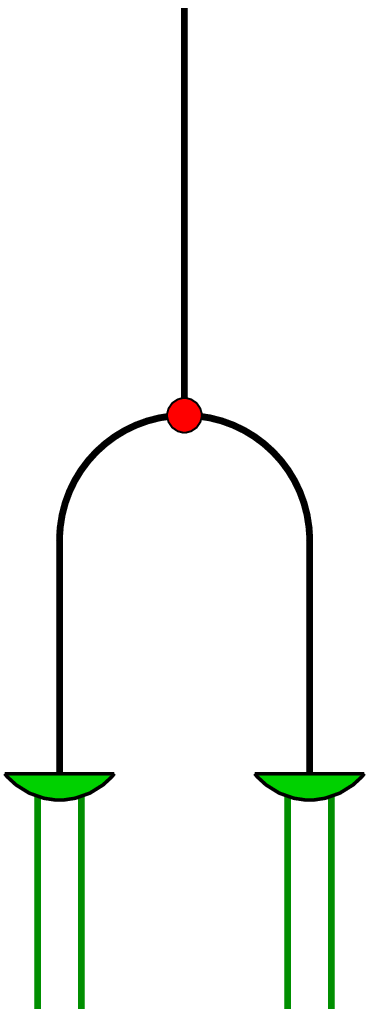}}
  \put(-4,-8)   {\scriptsize$ U^{\!\vee} $}
  \put(7,-8)    {\scriptsize$ U $}
  \put(-6.8,25) {\scriptsize$ \iota_{\!U}^{} $}
  \put(40.8,25) {\scriptsize$ \iota_{\!V}^{} $}
  \put(17.5,115){\scriptsize$ \H $}
  \put(22.6,68.3) {\scriptsize$ m_\H^{} $}
  \put(25,-8)   {\scriptsize$ V^{\!\vee} $}
  \put(36,-8)   {\scriptsize$ V $}
  \put(60,50)   {$ := $}
  \put(95,0) { {\Includepichtft{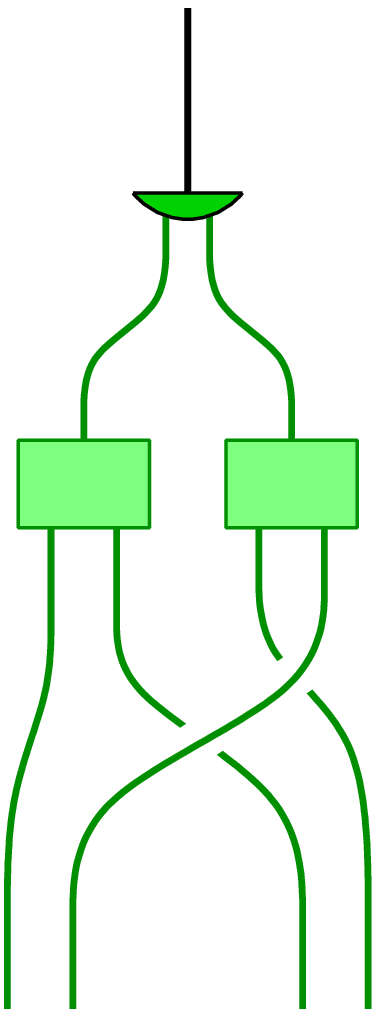}}
  \put(-6,-8)   {\scriptsize$ U^{\!\vee} $}
  \put(6,-8)    {\scriptsize$ U $}
  \put(-11,68)  {\scriptsize$ \gamma_{U,V}^{} $}
  \put(27.8,89) {\scriptsize$ \iota_{V\!\otimes U}^{} $}
  \put(33.8,66) {\scriptsize$ \id_{V\!\otimes U} $}
  \put(17.5,115){\scriptsize$ \H $}
  \put(27,-8)   {\scriptsize$ V^{\!\vee} $}
  \put(39,-8)   {\scriptsize$ V $}
  \put(19.5,24.3) {\scriptsize$ c $}
  \put(30,31) {\scriptsize$ c $} 
   }
    \put(235,0) {
  \put(0,4) { {\Includepichtft{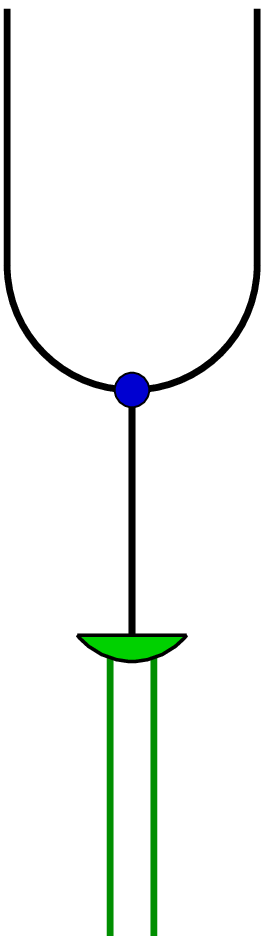}}
  \put(-2,106)  {\scriptsize$ \H $}
  \put(25.3,106){\scriptsize$ \H $}
  \put(6,-8)    {\scriptsize$ U^{\!\vee} $}
  \put(18,-8)   {\scriptsize$ U $}
  \put(15.4,54.2) {\scriptsize$ \Delta_\H $}
  \put(21,32) {\scriptsize$ \iota_{\!U}^{} $}
   }
  \put(55,50)   {$ := $}
  \put(85,0) { {\Includepichtft{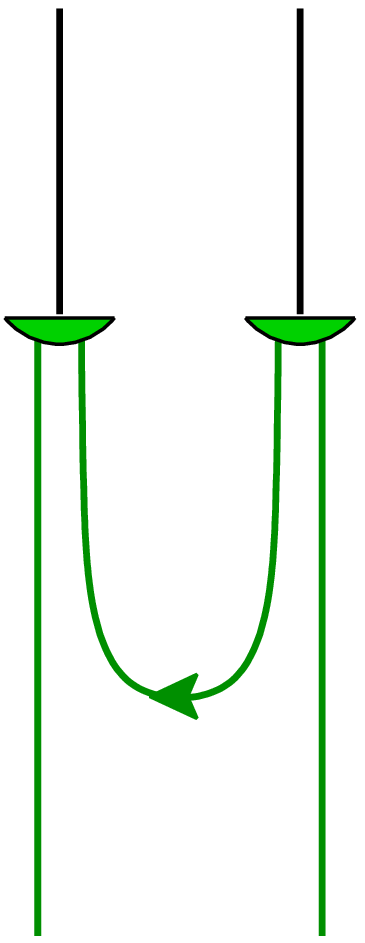}
  \put(3,106)  {\scriptsize$ \H $}
  \put(30,106)  {\scriptsize$ \H $}
  \put(0,-8)    {\scriptsize$ U^{\!\vee} $}
  \put(32,-8)   {\scriptsize$ U $}
  \put(-7.8,67) {\scriptsize$ \iota_{\!U}^{} $}
  \put(40.3,67) {\scriptsize$ \iota_{\!U}^{} $}}
   } }
   }
    \put(0,20) {
  \put(0,12) { {\Includepichtft{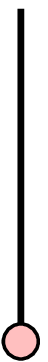}}
  \put(-.7,42) {\scriptsize$ \H $}
  \put(6,1)   {\scriptsize$ \eta_\H^{} $}
   }
  \put(25,25) {$ := $}
  \put(55,-5) {\Includepichtft{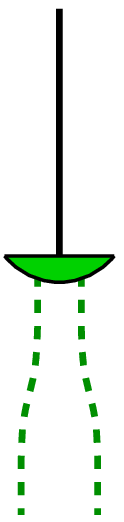}
  \put(3.3,59){\scriptsize$ \H $}
  \put(13.3,27) {\scriptsize$ \iota_{\!\one}^{} $}}
  \put(135,0) { {\Includepichtft{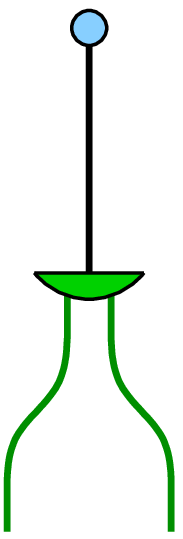}}
  \put(-3,-8) {\scriptsize$ U^{\!\vee} $}
  \put(16,-8) {\scriptsize$ U $}
  \put(13.4,54.5) {\scriptsize$ \eps_\H^{} $}
  \put(16.9,27.5) {\scriptsize$ \iota_{\!U}^{} $}
   }
  \put(174,30){$ := $}
  \put(203,0) {\Includepichtft{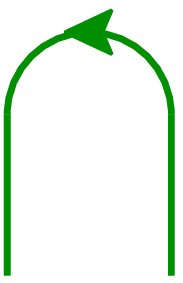}
  \put(-3,-8) {\scriptsize$ U^{\!\vee} $}
  \put(16,-8) {\scriptsize$ U $}}
   }
  \put(290,-7) { {\Includepichtft{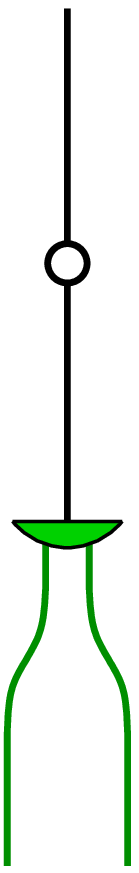}}
  \put(-3,-8) {\scriptsize$ U^{\!\vee} $}
  \put(11,-8) {\scriptsize$ U $}
  \put(11.4,65) {\scriptsize$ \apo_\H $}
  \put(4.6,99)  {\scriptsize$ \H $}
  \put(14.4,37) {\scriptsize$ \iota_{\!U}^{} $}
   }
  \put(329,37){$ := $}
  \put(358,-7) {\Includepichtft{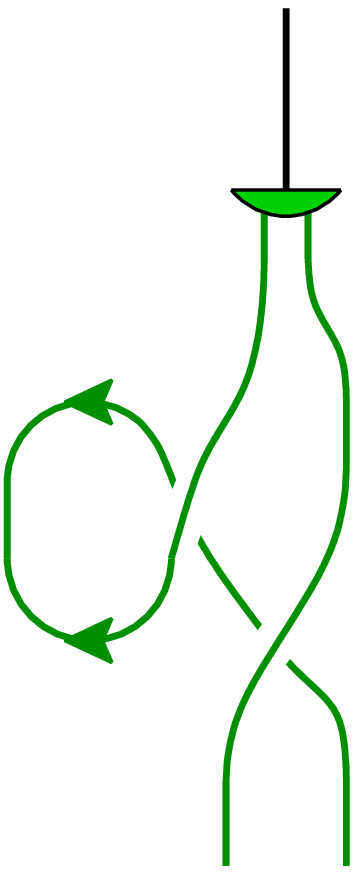}
  \put(21,-8) {\scriptsize$ U^{\!\vee} $}
  \put(35,-8) {\scriptsize$ U $}
  \put(28.4,99) {\scriptsize$ \H $}
  \put(38.4,73) {\scriptsize$ \iota_{\!U^\vee_{}}^{} $}
  \put(22,38.1) {\scriptsize$ c $}
  \put(28.5,17) {\scriptsize$ c $} }
  }
\H\ also has a two-sided integral and a Hopf pairing. 
For semisimple modular \C, the coend \H\ is given by
$\H \,{=}\, \bigoplus_{i\in\IJ} S_i^\vee \oti S_i^{} \,{\in}\, \C$.

Morphism spaces of the form $\Hom(\H^g,V_1\oti\cdots\oti V_n)$
carry \cite{lyub6} natural representations of the mapping class group
$\Gamma_{\!g,n}$ of Riemann surfaces of genus $g$ with $n$ marked points.
We therefore call \H\ the \emph{\chirhh}.


\subsection{Modules over the \chirhh}

For any object $V$ of \C, the family of morphisms from $U^\vee \oti U\oti V$
to $\H \oti V$ on the right hand side of
   \eqpic{QHX} {100} {46} { \put(0,8){
   \put(0,0)  {\begin{picture}(0,0)(0,0)
                   \scalebox{.38}{\includegraphics{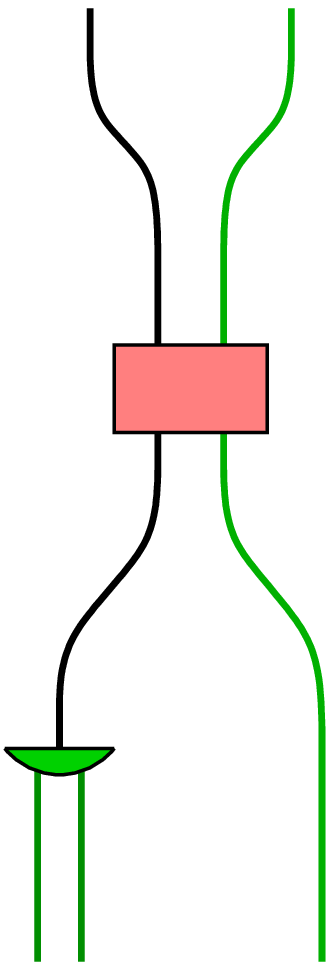}}\end{picture}
   \put(-15,60)  {$ \Qq_{\H,V} $}
   \put(-4,-8.5) {\scriptsize$ U^{\!\vee} $}
   \put(6.8,109) {\scriptsize$ \H $}
   \put(7,-8.5)  {\scriptsize$ U $}
   \put(31.3,-8.5){\scriptsize$ V $}
   \put(30.4,109){\scriptsize$ V $}
   \put(13.9,22.5) {\scriptsize$ \iota_{\!U}^{} $}
   }
   \put(60,50)   {$ := $}
   \put(94,0) {\begin{picture}(0,0)(0,0)
                   \scalebox{.38}{\includegraphics{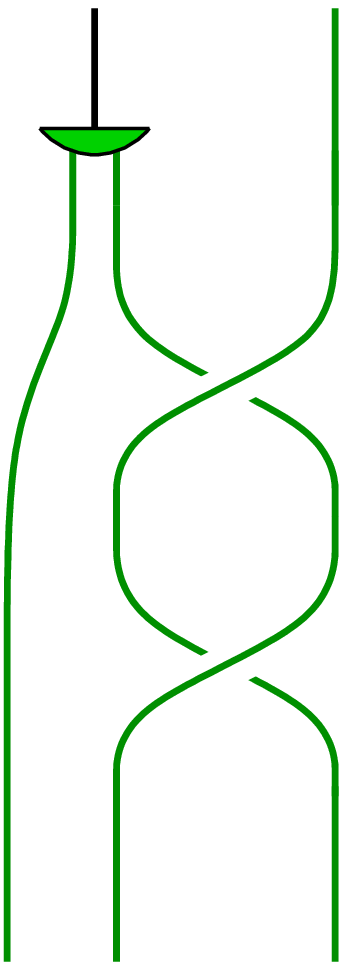}}\end{picture}
   \put(-6,-8.5) {\scriptsize$ U^{\!\vee} $}
   \put(7.8,109) {\scriptsize$ \H $}
   \put(8,-8.5)  {\scriptsize$ U $}
   \put(22.2,27.3) {\scriptsize$ c $}
   \put(22.2,58.4) {\scriptsize$ c $}
   \put(34,-8.5) {\scriptsize$ V $}
   \put(34.3,109){\scriptsize$ V $}
   \put(17.5,91) {\scriptsize$ \iota_{\!U}^{} $}
   } } }
is dinatural in the first two arguments and thus defines a morphism $\Qq_{\H,V}$ 
in $\EndC(\H\oti V)$, which we call the \emph{partial monodromy} of $V$ with 
respect to \H. Composition with the counit of \H\ supplies a morphism
  \be
  \rho^\H_V := (\eps_\H^{} \oti \id_V) \circ \Qq_{\H,V} ~\in\, 
  \mathrm{Hom}_\C(\H\oti V,V) \,.
  \label{def:rhoHV}
  \ee

\begin{lemma}\label{Lemma:rhoHV}
The morphism \eqref{def:rhoHV} endows the object $V$ of \C\ with the structure
of an \H-module internal to \C.
\end{lemma}

\begin{proof}
Unitality follows directly from the definition of the unit of \H.
Compatibility with the product of \H\ reduces to an application of the
defining properties of the braiding:
  \eqpic{Qrep} {420} {55} {
  \put(0,0)  {\begin{picture}(0,0)(0,0)
        \scalebox{.38}{\includegraphics{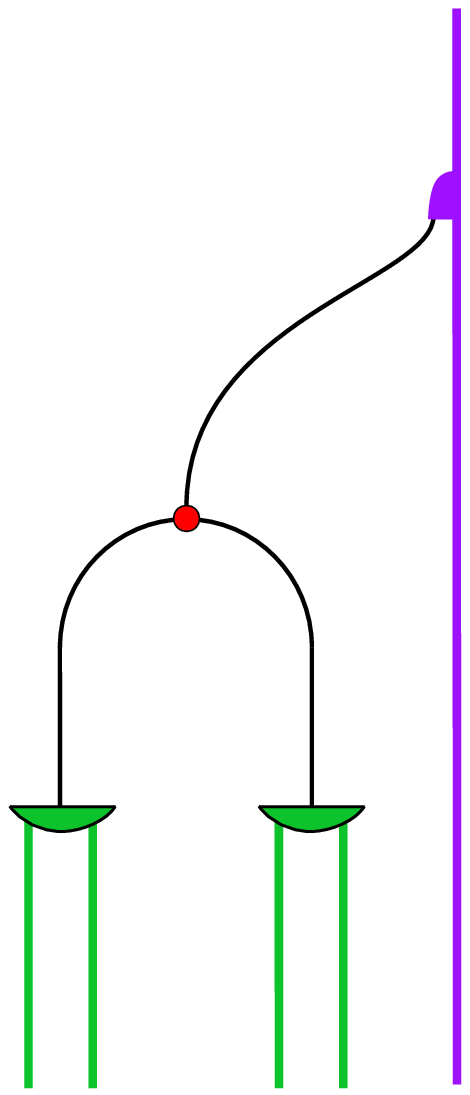}}\end{picture}
  \put(-5,-8.5) {\scriptsize$X^{\!\vee}$}
  \put(7,-8.5)  {\scriptsize$X$}
  \put(24,-8.5) {\scriptsize$Y^{\!\vee}$}
  \put(35,-8.5) {\scriptsize$Y$}
  \put(46,-8.5) {\scriptsize$V$}
  \put(46.4,123)  {\scriptsize$V$}
  \put(22,64.5) {\scriptsize$m_\H$}
  \put(34,97)   {\scriptsize$\rho^\H_V$}
  \put(-7.5,28.6) {\scriptsize$\iota_{\!X}^{}$}
  \put(21.5,28.6) {\scriptsize$\iota_{\!Y}^{}$}
  }
  \put(80,55)   {$ = $}
  \put(110,0)  {\begin{picture}(0,0)(0,0)
        \scalebox{.38}{\includegraphics{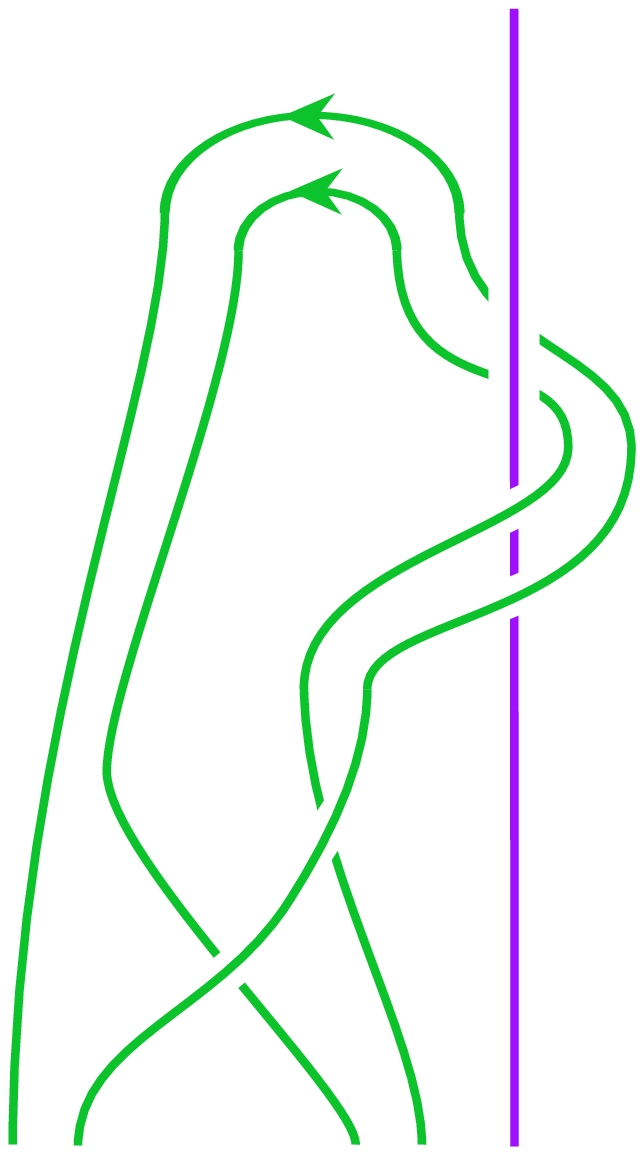}}\end{picture}
  \put(-6,-8.5) {\scriptsize$X^{\!\vee}$}
  \put(6,-8.5)  {\scriptsize$X$}
  \put(33,-8.5) {\scriptsize$Y^{\!\vee}$}
  \put(44,-8.5) {\scriptsize$Y$}
  \put(53,-8.5) {\scriptsize$V$}
  \put(53,128)  {\scriptsize$V$}
  \put(21.5,13.5) {\scriptsize$c$}
  \put(37.5,33.5) {\scriptsize$c$}
  \put(57.5,55.5) {\scriptsize$c$}
  \put(49.5,71.5) {\scriptsize$c$}
  \put(49.5,79.5) {\scriptsize$c$}
  \put(58.5,91.5) {\scriptsize$c$}
  }
  \put(195,55)   {$ = $}
  \put(230,0)  {\begin{picture}(0,0)(0,0)
        \scalebox{.38}{\includegraphics{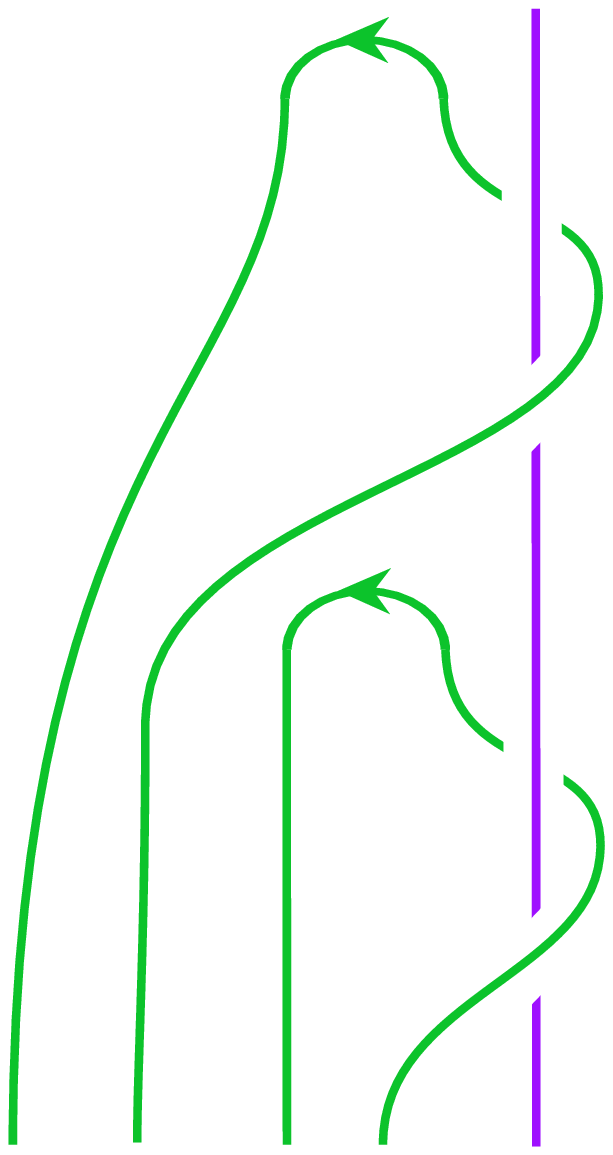}}\end{picture}
  \put(-6,-8.5) {\scriptsize$X^{\!\vee}$}
  \put(10,-8.5) {\scriptsize$X$}
  \put(26,-8.5) {\scriptsize$Y^{\!\vee}$}
  \put(39,-8.5) {\scriptsize$Y$}
  \put(54.9,-8.5) {\scriptsize$V$}
  \put(55.5,128){\scriptsize$V$}
  \put(60,17.5) {\scriptsize$c$}
  \put(60,43.5) {\scriptsize$c$}
  \put(60,78.5) {\scriptsize$c$}
  \put(60,104.5){\scriptsize$c$}
  }
  \put(315,55)   {$ = $}
  \put(350,0)  {\begin{picture}(0,0)(0,0)
        \scalebox{.38}{\includegraphics{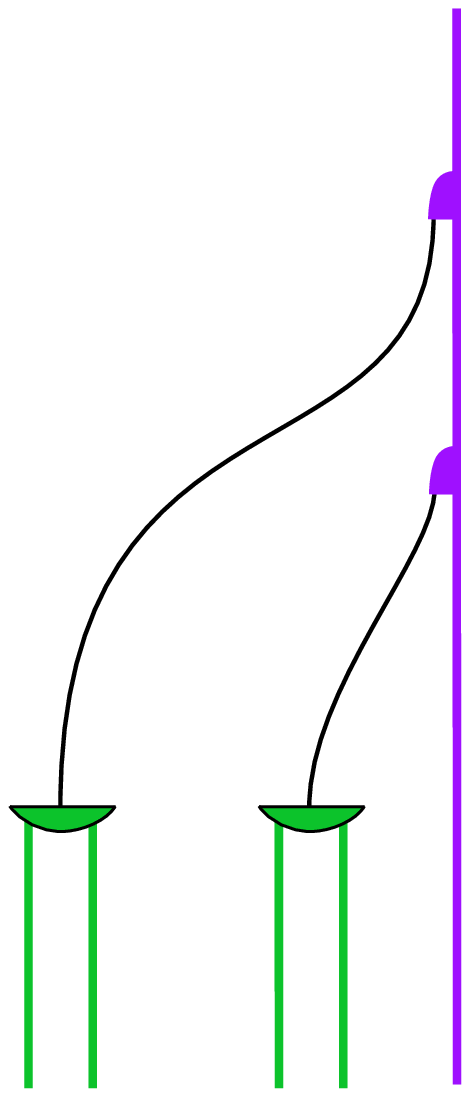}}\end{picture}
  \put(-5,-8.5) {\scriptsize$X^{\!\vee}$}
  \put(7,-8.5)  {\scriptsize$X$}
  \put(24,-8.5) {\scriptsize$Y^{\!\vee}$}
  \put(35,-8.5) {\scriptsize$Y$}
  \put(46,-8.5) {\scriptsize$V$}
  \put(46.3,123){\scriptsize$V$}
  \put(34,68)   {\scriptsize$\rho^\H_V$}
  \put(34,97)   {\scriptsize$\rho^\H_V$}
  \put(12,30)   {\scriptsize$\iota_{\!X}^{}$}
  \put(40,30)   {\scriptsize$\iota_{\!Y}^{}$}
  } }
Here the first equality combines the definitions of $\rho^\H_V$ in \eqref{def:rhoHV}
with those of $m_\H$ and $\eps_\H$ in \eqref{H_hopf}.
\end{proof}

The left action \eqref{def:rhoHV} of \H\ on an object $V$ of \C\ should not be confused
with the right \emph{co}action of \H\ on $V$ that is obtained \cite{lyub8} 
by combining the dinatural morphism $i_U$ with the coevaluation for $U$.
The latter only uses the duality of \C, whereas the former uses in addition the
braiding, or rather, the monodromy of \C.

For algebras in monoidal categories one can set up their representation theory in
a way very similar as for conventional \ko-algebras. If the category is sovereign,
one can in particular consider the character $\chii^\H_V $ of the \H-module
$(V,\rho^\H_V)$; it is given by
  \eqpic{Qrep2} {100} {44} {
  \put(-70,51) {$ \chii^\H_V \circ i_U ~= $}
  \put(0,0)  {\begin{picture}(0,0)(0,0)
                   \scalebox{.38}{\includegraphics{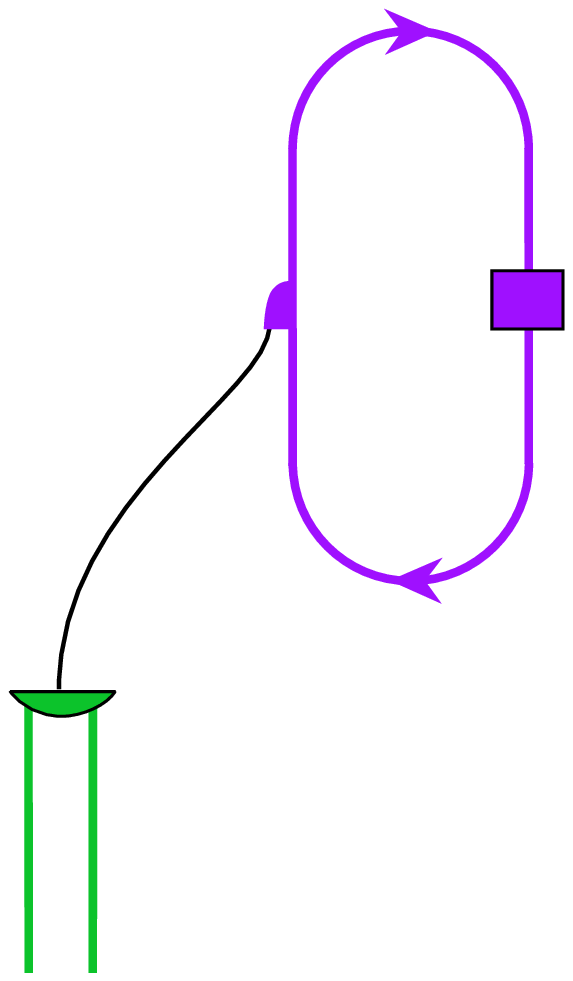}}\end{picture}
  \put(-5,-8.5){\scriptsize$U^{\!\vee}$}
  \put(7,-8.5) {\scriptsize$U$}
  \put(13,30)  {\scriptsize$\iota_{\!U}^{}$}
  \put(17.5,73){\scriptsize$\rho^\H_V$}
  \put(33,58)  {\scriptsize$V$}
  \put(63,73)  {\scriptsize$\pi_V$}
  }
  \put(81,51)  {$ = $}
  \put(110,0) {\begin{picture}(0,0)(0,0)
                   \scalebox{.38}{\includegraphics{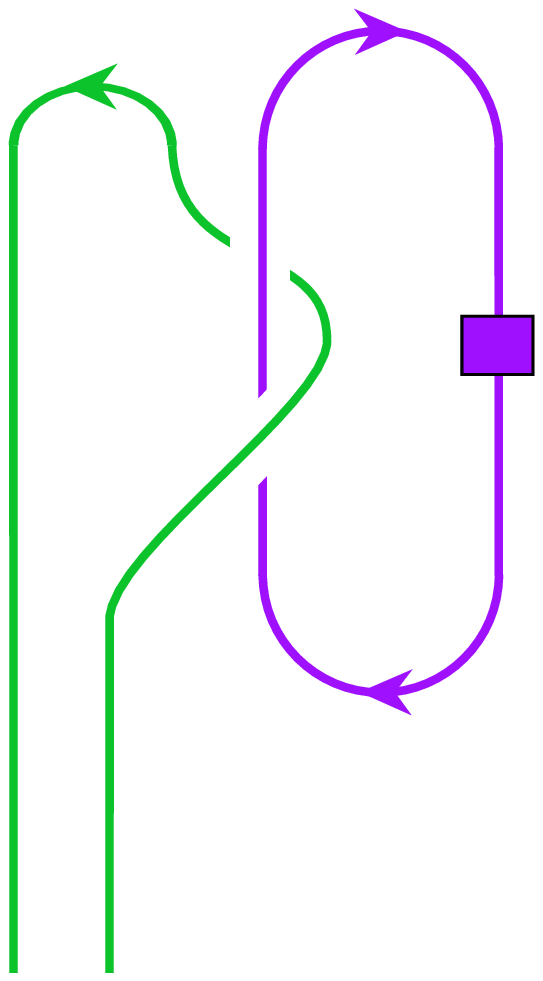}}\end{picture}
  \put(-6,-8.5){\scriptsize$U^{\!\vee}$}
  \put(7,-8.5) {\scriptsize$U$}
  \put(30,44)  {\scriptsize$V$}
  \put(21,59)  {\scriptsize$c$}
  \put(30,80)  {\scriptsize$c$}
  \put(60,68)  {\scriptsize$\pi_V$}
   } }
A new feature appearing here as compared to vector spaces is that the left
and right dual of an object of \C\ need not be equal, but only naturally
isomorphic. This necessitates the insertion of an appropriate isomorphism
$\pi_V \,{\in}\, \mathrm{Hom}_\C(V^\vee,{}^{\vee\!}V)$, 
forming part of a sovereign structure on \C.

Since \H\ is a Hopf algebra, there is a natural notion of dual module.
The character of the \H-module $V^\vee$ dual to $V$ turns out to be given
by essentially the same morphism as in \eqref{Qrep2}, except that the braidings
are replaced by inverse braidings.

\medskip

Now we specialize to the case $\C \,{=}\, \HMod$ for a factorizable Hopf algebra
$H$. In this case one can describe the coend \H\ explicitly 
\cite{lyub6,kerl5,vire4}: as an object of \C{}
it is the vector space \Hs\ endowed with the coadjoint left $H$-action, and the
morphisms of the dinatural family $i$ are given by
  \eqpic{pic-iU} {120} {37} {
  \put(0,39)    {$ i_U ~= $}
  \put(52,0) { \begin{picture}(0,0)(0,0)
       \scalebox{.38}{\includegraphics{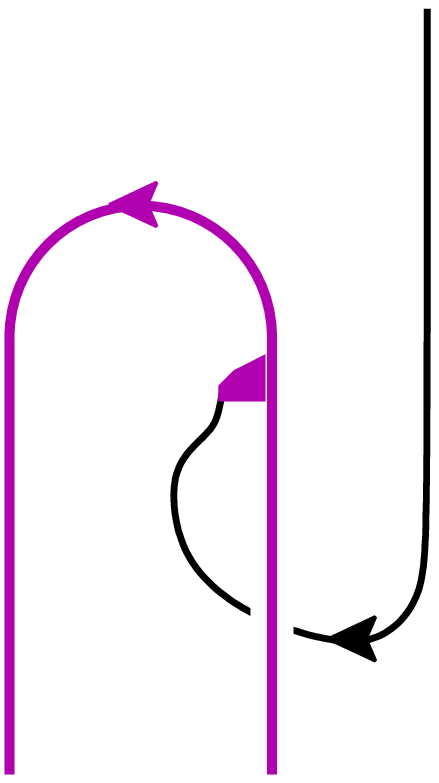}}\end{picture}
  \put(-4,-8.5) {\scriptsize$ U^* $}
  \put(25,-8.5) {\scriptsize$ U $}
  \put(31,43)   {\scriptsize$ \rho_{\!U}^{} $}
  \put(43,89)   {\scriptsize$ \Hss $}
  } }
Further, the monodromy appearing in \eqref{def:rhoHV} is now given by the action 
of the monodromy matrix $Q$ of $H$ on the tensor product $H$-module $U\oti V$, and
the sovereignty isomorphism is given by
  \eqpic{pic-piV} {175} {36} {
  \put(0,39)    {$ \pi_V ~= $}
  \put(53,0)  {\begin{picture}(0,0)(0,0)
       \scalebox{.38}{\includegraphics{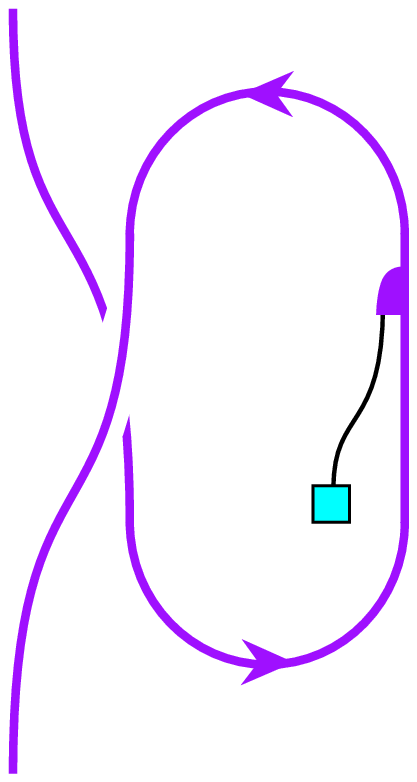}}\end{picture}
  \put(-4.5,-8) {\scriptsize$ V^*_{} $}
  \put(-3,88)   {\scriptsize$ V^*_{} $}
  \put(28,28)   {\scriptsize$ t $}
  \put(45,52)   {\scriptsize$ \rho_{\!V}^{} $}
  }
  \put(123,39)  {$\in \Endk(V^*_{})$}
  }
with $t \,{=}\, u\,v^{-1}$ the product of the Drinfeld element 
$u \,{:=}\, m \,{\circ}\, (\apo\oti\idH) \,{\circ}\, R_{21} \,{\in}\, H$
and the inverse of the ribbon element of $H$.

As a consequence we have the following description of \H-characters:

\begin{lemma}\label{lem:chiiH=drinchiiH}
The character of the \H-module $(M,\rho^\H_M)$ obeys
  \be
  \chii^\H_M = \chii^H_M \circ m \circ (t \oti \drin)
  = \chii^H_M \circ m \circ (\drin \oti t)
  \label{chiiH=drinchiiH}
  \ee
with
  \eqpic{def-drin} {220} {18} {
  \put(0,24) {$ \drin ~:=~ (b_H \oti \id_H) \circ (\idHs \oti Q) ~= $}
  \put(194,0)  { \begin{picture}(0,0)(0,0)
                   \scalebox{.38}{\includegraphics{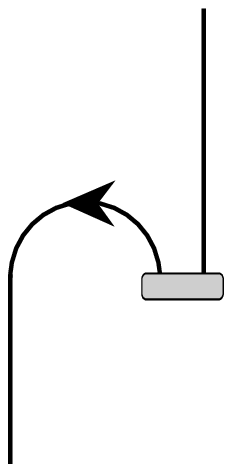}}\end{picture}
  \put(-5,-8.5){\scriptsize$\Hs$}
  \put(18,54)  {\scriptsize$H$}
  \put(26,18)  {\scriptsize$Q$}
  } }
$($and with $t$ regarded as an element of $\Homk(\ko,H)\,)$.
\end{lemma}

\begin{proof}
Inserting \eqref{pic-iU} and \eqref{pic-piV} into \eqref{Qrep2} and using that the
monodromy in \HMod\ is furnished by the action of the monodromy matrix $Q$, one obtains
  \eqpic{char-drin} {120} {35} {
  \put(0,40){$\chii^\H_M ~=$}
  \put(57,0) {\begin{picture}(0,0)(0,0)
                   \scalebox{.38}{\includegraphics{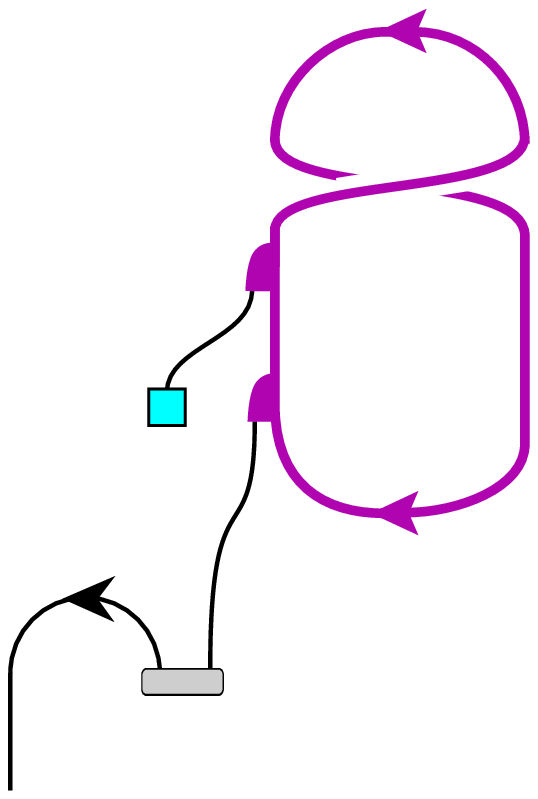}}\end{picture} 
  \put(-4.4,-8.5) {\scriptsize$ \Hs $}
  \put(10.8,40.5){\scriptsize$t$}
  \put(16.2,3.8){\scriptsize$Q$}
  \put(31,42.5) {\scriptsize$ \rho_{\!M}^{} $}
  \put(31,56.5) {\scriptsize$ \rho_{\!M}^{} $}
  \put(58,76)   {\scriptsize$M$}
  } }
Using the \rep\ property and comparing with \eqref{def_char} then yields
the first of the equalities \eqref{chiiH=drinchiiH}. The expression given by 
the second equality can be obtained in a similar way; that both expressions are 
valid is a consequence of the sphericity of the category of \ko-vector spaces.
\end{proof}

\begin{remark}
The mapping $\drin$ \eqref{def-drin} is called the \emph{Drinfeld map}.
A priori $\drin$ is just a linear map in $\Homk(\Hs,H)$, but actually
\cite[Prop.\,2.5(5)]{coWe4} it is a module morphism from \Hs\ with the
left coadjoint $H$-action to $H$ with the left adjoint action on itself.
\end{remark}


\subsection{Modules over the \bulkhh}\label{sec:bulkhh}

In CFT terms, what we have dealt with so far is a chiral half of the theory.
We now proceed from the chiral to the full theory. In the present setting this
means that we no longer work with the ribbon category \HMod\ of left $H$-modules, 
but now with the ribbon category
$\HBimod \,{\simeq}\, \overline{\HMod} \,{\boxtimes} \HMod$ of $H$-bimodules.
There are then \emph{two} coends of interest to us. The first is the
\emph{\bulkhh} \K. This is just the bimodule version of the coend \H, i.e.\
  \be
  \K = \coen {X\in\HBimod} X^\vee \oti X \,,
  \ee
where the bifunctor $\otimes\colon \HBimod \,{\times}\, \HBimod \To \HBimod$ is now
the tensor product in \HBimod. Explicitly \cite[App.\,A.4]{fuSs3}, \K\ is the
\emph{coadjoint bimodule}, i.e.\ the vector space $\Hs \otik \Hs$
endowed with the coadjoint left $H$-action on the first tensor factor and with
the coadjoint right $H$-action on the second factor, and the
dinaturality morphisms are given by \cite[(A.30)]{fuSs3}
  \eqpic{def_iHaa_X} {90} {44} {
  \put(0,44)     {$ i^\K_X ~= $}
  \put(53,0) {\begin{picture}(0,0)(0,0)
                   \scalebox{.38}{\includegraphics{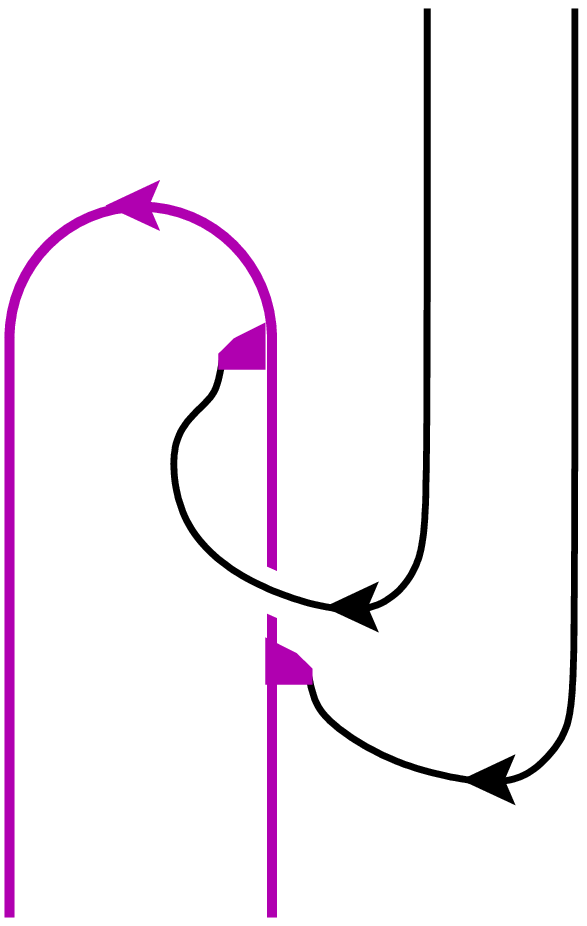}}\end{picture} 
  \put(-4.1,-8.5){\scriptsize$ X^{*_{}}_{\phantom:} $}
  \put(19.4,28)  {\scriptsize$ \ohr_{\!X}^{} $}
  \put(25.3,-8.5){\scriptsize$ X $}
  \put(31.2,62)  {\scriptsize$ \rho_{\!X}^{} $}
  \put(42.9,104) {\scriptsize$ \Hss $}
  \put(59.4,104) {\scriptsize$ \Hss $}
  } }
where $\rho_{\!X}$ and $\ohr_{\!X}$ are the left and right actions of $H$ on the
$H$-bimodule $X$ (compare the chiral version \eqref{pic-iU}).

For the characters of modules over the \bulkhh\ the following analogue of
Lemma \ref{lem:chiiH=drinchiiH} holds.

\begin{lemma}
The \K-character of a \K-module $(X,\rho^\K_X)$ can be expressed through
characters for the ordinary Hopf algebra $H\oti H\op$ as
  \be
  \chii^\K_X = \chii^{H\otimes H\op}_X \circ (m \oti m) \circ
  ( t \oti f_{Q^{-1}} \oti \drin \oti t) \,.
  \label{chiKX-fQi-fQ}
  \ee
\end{lemma}

\begin{proof}
Inserting \eqref{def_iHaa_X} into the formula \eqref{Qrep2} for the character,
as adapted to the present situation (i.e.\ in particular with \H\ replaced by \K\
and with the monodromy the one of \HBimod), it follows that
  \eqpic{KX_char} {240} {39} {
  \put(-52,41)   {$ \chii^\K_X = $}
  \put(0,0)  {\Includepichtft{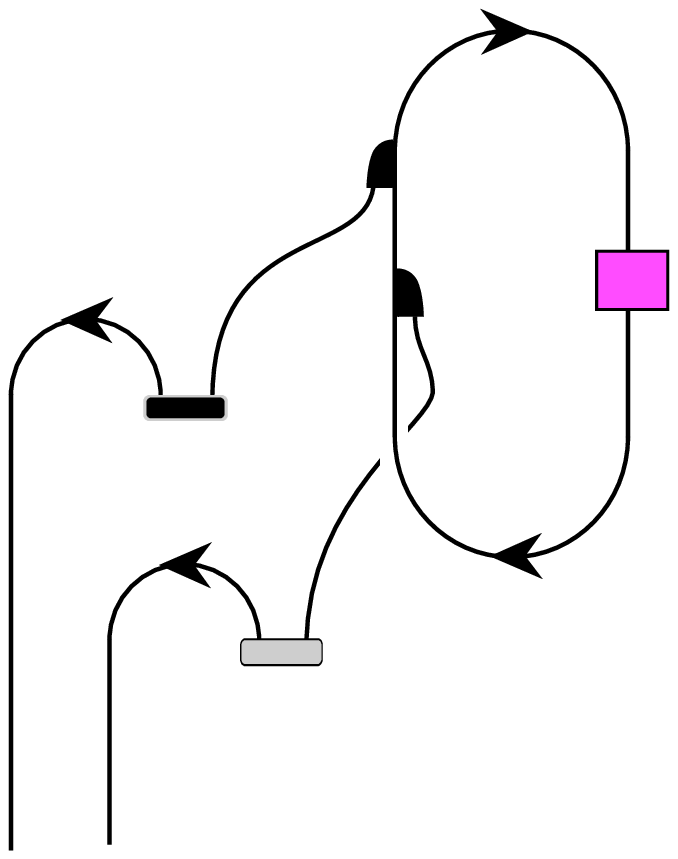}
  \put(12.3,40.7){\scriptsize$Q^{-1}$}
  \put(26.7,13.3){\scriptsize$Q$}
  \put(-6,-8.5)  {\scriptsize$H^*$}
  \put(7,-8.5)   {\scriptsize$H^*$}
  \put(29.3,75.2){\scriptsize$ \rho_{\!X}^{} $}
  \put(45.6,62.2){\scriptsize$ \ohr_{\!X}^{} $}
  \put(74,61.5)  {\scriptsize $\pi_X^{}$}
  }
  \put(95,41)    {$ = $}
  \put(125,0) {\Includepichtft{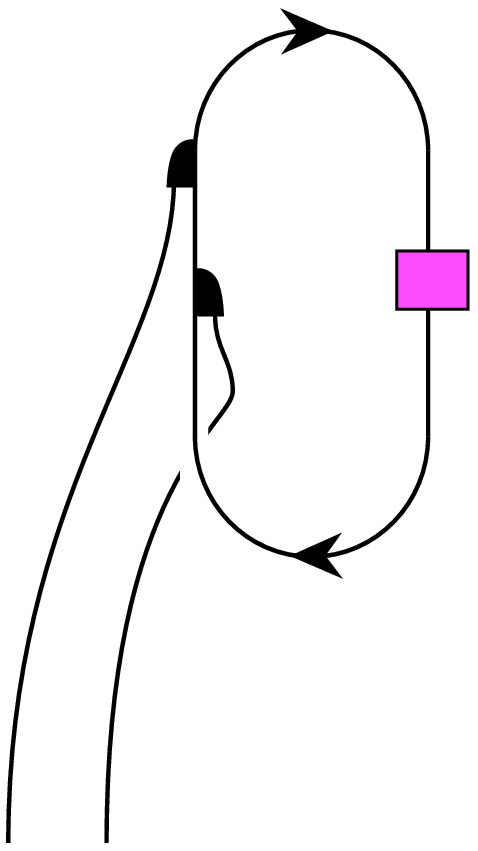}
  \put(-4.5,-8.5)  {\scriptsize$ H $}
  \put(7,-8.5)   {\scriptsize$ H $}
  \put(8.3,74.5) {\scriptsize$ \rho_{\!X}^{} $}
  \put(24.4,61.3){\scriptsize$ \ohr_{\!X}^{} $}
  \put(53,61.5)  {\scriptsize $\pi_X^{}$}
  }
  \put(200,41)   {$ \circ\;\;(f_{Q^{-1}}\oti \drin)\,, $}
  }
where $\drin$ is the Drinfeld map \eqref{def-drin}
and $f_{Q^{-1}}$ is the analogous morphism with $Q$ replaced by $Q^{-1}$.
 \\
The sovereignty morphism $\pi_X$ is given by the bimodule analogue of 
\eqref{pic-piV}, i.e.\ with the left action by the element $t\,{\in}\, H$ 
complemented with a right action by $t$ \cite[(4.8)]{fuSs3}. Each of the two 
occurrences of $t$ can be manipulated in the same way as the single $t$
in Lemma \ref{lem:chiiH=drinchiiH}. Regarding the $H$-bimodule $X$ as a left
$H\oti H\op$-module, this yields \eqref{chiKX-fQi-fQ}.
\end{proof}

The result \eqref{chiKX-fQi-fQ} is in fact not so surprising, because 
there is a ribbon equivalence between the categories of 
$H\,{\otimes}\,H$-modules and $H$-bimodules (the equivalence functor is 
given in formula (A.22) of \cite{fuSs3}), and this equivalence maps the 
$H\,{\otimes}\,H$-module $\H \otik \H$ to the bimodule \K. Furthermore,
invoking also Lemma \ref{lem:chiiH=drinchiiH} and formula
\eqref{chiABX}, we arrive at a \emph{chiral decomposition} of \K-characters:

\begin{proposition}\label{prop:chiKX}
The \K-character of a \K-module $X$ can be expressed through
characters for the \chirhh\ in the form
  \be
  \chii^\K_X = \sum_{i,j\in\IJ} n_{i,j}\, \chii^\H_{\overline i} \oti \chii^\H_j \,,
  \label{chiKX-chiL-chiL}
  \ee
where $\chii^\H_i$ is the \H-character that via \eqref{chiiH=drinchiiH} corresponds
to the irreducible $H$-character $\chii^H_i$,
$ \chii^\H_{\overline i}$ is the \H-character of the corresponding dual \H-module,
and $n_{i,j}$ $(i,j\,{\in}\,\IJ)$ are
the non-negative integers that appear in formula \eqref{chiABX}.
\end{proposition}

\begin{proof}
Manipulating the sovereignty isomorphism in \eqref{KX_char}
in the same way as in the proof of Lemma \ref{lem:chiiH=drinchiiH} and invoking 
\eqref{chiABX} for $A \,{=}\, H$ and $B \,{=}\, H\op$ as well as
the equivalence between $H$-bimodules and $H\oti H\op$-modules, we arrive at
  \be
  \chii^\K_X = \sum_{i,j\in\IJ} n_{i,j}\,
  \big[ \chii^H_i \,{\circ}\, m \,{\circ}\, (t \oti f_{Q^{-1}}) \big] \otimes
  \big[ \chii^H_j \,{\circ}\, m \,{\circ}\, (\drin \oti t) \big] \,.
  \label{chiKX=nij..}
  \ee
By \eqref{chiiH=drinchiiH}, the second tensor factor equals $\chii^\H_j$.
For the first factor, the presence of $f_{Q^{-1}}$ instead of $\drin$ amounts
to replacing the braiding by the opposite braiding in the \rep\ morphism, and 
thus (compare the corresponding remark after \eqref{Qrep2})
\eqref{chiiH=drinchiiH} gives again an \H-character, but now for the dual 
module. Together this yields \eqref{chiKX-chiL-chiL}.
\end{proof}


\subsection{The character of the bulk Frobenius algebra}

The second coend we need is the one for the functor from 
$\HMod\op \,{\times}\, \HMod$ to \HBimod\ that on objects acts as 
$(U,V) \,{\mapsto}\, U^\vee \,{\boxtimes}\, V$; we denote it by
  \be
  \F = \coen {U\in\HMod} U^\vee \,{\boxtimes}\, U \,.
  \ee
As already suggested by the chosen notation, \F\ is nothing but the 
coregular bimodule, i.e.\ the \emph{bulk Frobenius algebra} featuring in
Theorem \ref{thm1}, with dinatural family coinciding, as linear maps,
with the one of the \chirhh\ in \eqref{pic-iU}
(for details see \cite[App.\,A.2]{fuSs3}).

\medskip

We have now collected all ingredients for establishing Theorem \ref{thm2}.
\\[5pt]
\emph{Proof of Theorem} \ref{thm2}.
\\
According to \eqref{F-H-bimod-char}, for $X \,{=}\, \F$ the coefficients
$n_{i,j}$ in \eqref{chiABX}, and thus those in 
\eqref{chiKX=nij..}, are given by the entries $c_{i,j}$ of 
the Cartan matrix of \HMod. Using that $\overline{\overline i} \,{=}\, i$,
we thus arrive at \eqref{eq:thm2}.
\hfill{$\Box$}


\section{Outlook}

The quest for a classification of modular invariant partition functions has 
been an important activity in mathematical physics in the late 1980s and 
early 1990s. Nowadays it may be considered as
superseded by approaches based on category-theoretic tools.
In retrospect it is surprising how far one could get in this quest by imposing
only a few convenient necessary conditions. The result by Max Kreuzer and Bert 
Schellekens \cite{krSc} is still the best available and, most probably, the 
best possible \emph{systematic} result of this activity.

The Kreuzer-Schellekens classification also played a central role in 
developments that led to the modern more mathematical approach to rational
conformal field theory \cite{fuRs4,fjfrs}. This approach had in particular 
to reproduce their beautiful result, and indeed \cite{fuRs9} it does.

Today, one important activity is concerned with logarithmic conformal
field theories, which amounts to dropping the condition of semisimplicity.
For such theories, the only systematic information about torus partition
functions seems to be the one about the Cardy-Cartan invariant discussed here
(together with some automorphism-twisted versions \cite[Sect.\,6]{fuSs3}).
It is therefore encouraging that simple current symmetries, which are a crucial
input for the Kreuzer-Schellekens result, appear to occur in logarithmic CFT
\cite[Rem.\,5.3.2]{fhst} as well. Moreover, module categories over non-semisimple 
tensor categories which only have invertible simple objects have been classified
under certain assumptions \cite{gaMom}. It is quite reasonable to expect that 
these results can be combined with the structural insight in the 
Kreuzer-Schellekens formula to give a handle on partition functions of 
logarithmic CFTs that are not of the Cardy-Cartan form. But this extension of 
Max' ideas still awaits its realization.


 \vskip 2.5em
\noindent{\sc Acknowledgments:}
We thank Rolf Farnsteiner for very helpful discussions.
 \\
JF is largely supported by VR under project no.\ 621-2009-3993.
CSc is partially supported by the Collaborative Research Centre 676 
``Particles, Strings and the Early Universe
- the Structure of Matter and Space-Time''
and by the DFG Priority Programme 1388 ``Representation Theory''.

\newpage

 \newcommand\wb{\,\linebreak[0]} \def\wB {$\,$\wb}
 \newcommand\Bi[2]    {\bibitem[#2]{#1}}
 \newcommand\inBo[9]  {{\em #9}, in:\ {\em #1}, {#2}\ ({#3}, {#4} {#5}), p.\ {#6--#7} }
 \newcommand\inBO[9]  {{\em #9}, in:\ {\em #1}, {#2}\ ({#3}, {#4} {#5}), p.\ {#6--#7} {{\tt [#8]}}}
 \newcommand\J[7]     {{\em #7}, {#1} {#2} ({#3}) {#4--#5} {{\tt [#6]}}}
 \newcommand\JO[6]    {{\em #6}, {#1} {#2} ({#3}) {#4--#5} }
 \newcommand\JP[7]    {{\em #7}, {#1} ({#3}) {{\tt [#6]}}}
 \newcommand\BOOK[4]  {{\em #1\/} ({#2}, {#3} {#4})}
 \newcommand\PhD[2]   {{\em #2}, Ph.D.\ thesis #1}
 \newcommand\Prep[2]  {{\em #2}, preprint {\tt #1}}
 \def\adma  {Adv.\wb Math.}
 \def\aspm  {Adv.\wb Stu\-dies\wB in\wB Pure\wB Math.}
 \def\amjm  {Amer.\wb J.\wb Math.}
 \def\atmp  {Adv.\wb Theor.\wb Math.\wb Phys.}
 \def\coma  {Con\-temp.\wb Math.}
 \def\comp  {Com\-mun.\wb Math.\wb Phys.}
 \def\duke  {Duke\wB Math.\wb J.}
 \def\ijmp  {Int.\wb J.\wb Mod.\wb Phys.\ A}
 \def\jhep  {J.\wb High\wB Energy\wB Phys.}
 \def\joal  {J.\wB Al\-ge\-bra}
 \def\jopa  {J.\wb Phys.\ A}
 \def\joms  {J.\wb Math.\wb Sci.}
 \def\jktr  {J.\wB Knot\wB Theory\wB and\wB its\wB Ramif.}
 \def\jpaa  {J.\wB Pure\wB Appl.\wb Alg.}
 \def\jram  {J.\wB rei\-ne\wB an\-gew.\wb Math.}
 \def\momj  {Mos\-cow\wB Math.\wb J.}
 \def\nupb  {Nucl.\wb Phys.\ B}
 \def\pams  {Proc.\wb Amer.\wb Math.\wb Soc.}
 \def\plms  {Proc.\wB Lon\-don\wB Math.\wb Soc.}
 \def\ruma  {Revista de la Uni\'on Matem\'atica Argentina}
 \def\sema  {Selecta\wB Mathematica}
 \def\slnm  {Sprin\-ger\wB Lecture\wB Notes\wB in\wB Mathematics}
 \def\taac  {Theo\-ry\wB and\wB Appl.\wb Cat.}
 \def\taia  {Top\-o\-lo\-gy\wB and\wB its\wB Appl.}
 \def\trgr  {Trans\-form.\wB Groups}

\end{document}